\documentclass[10pt]{article}

\usepackage[a4paper, top=2.54cm, bottom=2.54cm, left=1.9cm, right=1.9cm]{geometry}

\usepackage{amsfonts, latexsym, amsmath, amssymb, fancyhdr, amsthm, graphicx, textcomp}
\usepackage{caption}
\usepackage{xcolor}%
\usepackage{enumitem}
\usepackage{graphics} 
\usepackage{epstopdf} 
\usepackage{epsfig} 
\usepackage[english]{babel}
\usepackage{setspace}
\usepackage[normalem]{ulem}
\usepackage{graphicx}
\usepackage{multicol}
 \usepackage{lineno}

\usepackage{lipsum}
\usepackage{multicol}
\usepackage{graphicx}
\usepackage{amsmath,amssymb}
\usepackage{ulem,xpatch}
\usepackage{graphicx} 
\usepackage{appendix}
\usepackage{mathtools}
\usepackage{dsfont}
\usepackage{amsfonts,latexsym}
\usepackage{ragged2e}
\usepackage{float}
\usepackage{verbatim}
\usepackage{wasysym}   
\usepackage{caption}
\usepackage{comment}
\usepackage{epsfig}
\usepackage{amsthm}
\usepackage{enumitem}
\usepackage[FIGTOPCAP]{subfigure}
\usepackage{scalefnt}
\usepackage[figuresright]{rotating}
\usepackage{hyperref} 
\usepackage{cite}  

\newenvironment{Figure}
  {\par\medskip\noindent\minipage{\linewidth}}
  {\endminipage\par\medskip}
\newtheorem{theorem}{Theorem}[section]

\newtheorem{lemma}[theorem]{Lemma}
\newtheorem{proposition}{Proposition}

\theoremstyle{definition}

\newtheorem{remark}{Remark}

\title{Modeling the mechanisms of antibody mixtures in viral infections: the cases of sequential homologous and heterologous dengue infections}

\author{
Charlotte Dugourd-Camus$^{1}$, Claudia P. Ferreira$^{2}$, Mostafa Adimy$^{1}$}

\date{}

\begin{document}
	
\maketitle
	
	\begin{center}
$^{1}$Inria, ICJ UMR5208, CNRS, Ecole Centrale de Lyon, INSA Lyon, Universit\'e Claude Bernard Lyon 1,
Universit\'e Jean Monnet, 69603 Villeurbanne, France\\
 \end{center}
	\begin{center}
		$^{2}$S\~{a}o Paulo State University (UNESP), Institute of Biosciences, 18618-689 Botucatu, SP, Brazil
	\end{center}
\begin{abstract}
Antibodies play an essential role in the immune response to viral infections, vaccination, or antibody therapy. Nevertheless, they can be either protective or harmful during the immune response. Moreover, competition or cooperation between mixed antibodies can enhance or reduce this protective or harmful effect. Using the laws of chemical reactions, we propose a new approach to modeling the antigen-antibody complex activity. The resulting expression covers not only purely competitive or purely independent binding but also synergistic binding which, depending on the antibodies, can promote either neutralization or enhancement of viral activity. We then integrate this expression of viral activity in a within-host model and investigate the existence of steady-states and their asymptotic stability. We complete our study with numerical simulations to illustrate different scenarios: firstly, where both antibodies are neutralizing, and secondly, where one antibody is neutralizing and the other enhancing. The results indicate that efficient viral neutralization is associated with purely independent antibody binding, whereas strong viral activity enhancement is expected in the case of purely competitive antibody binding.  Finally, data collected during a secondary dengue infection were used to validate the model. The data set includes sequential measurements of virus and antibody titers during viremia in patients. Data fitting shows that the two antibodies are in strong competition, as the synergistic binding is low. This contributes to the high levels of virus titers and may explain the Antibody-Dependent Enhancement phenomenon. Besides, the mortality of infected cells is almost twice as high as that of susceptible cells, and the heterogeneity of viral kinetics in patients is associated with variability in antibody responses between individuals. Other applications of the model may be considered, such as the efficacy of vaccines and antibody-based therapies.
	\\
	\\
	\textbf{Subject}: Biomathematics	
34C60, 34D05, 37N25, 92C40, 92C45.
\\
\\
\textbf{Author for correspondence:}
Charlotte Dugourd-Camus:
charlotte.dugourd@inria.fr\\ 
Claudia P. Ferrreira:
claudia.pio@unesp.br\\
Mostafa Adimy: 
mostafa.adimy@inria.fr
\end{abstract}

\section{Introduction}
\begin{multicols}{2}
Active immunity is triggered by the entrance of a pathogen or by vaccination and involves the production of antibodies by the immune system. On the other hand, passive immunity involves the vertical transmission of antibodies from mother to offspring or the horizontal transmission of antibodies from humans or animals to susceptible individuals through antibody donation (antibody therapy). The interaction between the antibody and the antigen occurs through a specific chemical reaction, to form an antigen-antibody complex. This immune complex 
can either neutralize or enhance the activity of the pathogen.  The strength of the antigen-antibody complex depends firstly on the affinity of the antibody for the antigen, secondly on the number of antigen-antibody binding sites, and thirdly on the structural arrangement of the interacting parts, \cite{einav_when_2020, santillan_use_2008}. In the case of secondary infection by a homologous or heterologous virus, pre-existing antibodies and newly produced antibodies may coexist, \cite{st_john_adaptive_2019}. In most cases, an active infection leads to the production of neutralizing antibodies. However, cross-reaction of pre-existing antibodies may or may not help neutralize this active infection; in fact, it has been observed that cross-reacting antibodies can enhance the infection, \cite{clapham_modelling_2016, estofolete_Plos_2023, reich_interactions_2013, tricou_kinetics_2011, wen_antibody_2020}. In the case of influenza, dengue, or COVID-19, for example, it has been observed that two successive homologous infections can increase the neutralizing effect of antibodies, \cite{einav_when_2020, puschnik_correlation_2013}. On the other hand, secondary infection with a heterologous dengue serotype may increase the risk of developing a severe form of the disease, due to Antibody-Dependent Enhancement (ADE) phenomenon, \cite{clapham_modelling_2016, estofolete_Plos_2023, goncalvez_monoclonal_2007, halstead_dengue_2014, parren_antiviral_2001, reich_interactions_2013, tricou_kinetics_2011, wen_antibody_2020}. Furthermore, in the presence of a mixture of two antibodies and a pathogen, two other factors come into play: (i) the competition between antibodies to bind to the antigen, and (ii) the synergistic interactions between antibodies, \cite{einav_when_2020,sanna_synergistic_2000}. This leads to three distinct types of antigen-antibody binding: (a) purely independent binding, where the binding of one antibody does not affect the binding of other antibodies to the same receptor; (b) purely competitive binding, where the two antibodies cannot bind simultaneously to the same receptor; and (c) an intermediate situation, where the two antibodies can bind simultaneously to many sites of the same receptor, with synergistic interactions, i.e., the binding of one has an effect on the binding of the other, either by making it weaker or stronger, \cite{einav_when_2020,sanna_synergistic_2000}.

Several mathematical within-host infectious disease models have been developed to study the dynamics of viral infections and their interaction with the immune system, \cite{Anam2024, ansari_within_2012, ben-shachar_minimal_2015, ceron_simple_2018, clapham_within_2014, gujarati_virus_2014, nuraini_within_2009, pawelek_modeling_2012, perelson_modelling_2002, Rubio2022, tang_modelling_2020, wodarz_mathematical_2002, Xu_Math_2024}. They involve, according to each paper, healthy and infected target cells, intracellular pathogen replication, T-cells, B-cells, cytokine production, and antibodies. Most of them resemble epidemiological models and use the mass action law or a saturation function to model the force of infection. In general, only one of the immune responses - humoral or cellular - is modeled, and very few models address interaction among components of the immune response. As an example, \cite{ansari_within_2012, clapham_within_2014, nuraini_within_2009} took into account only one component of the immune system: T-cells, the only ones capable of killing infected cells. In particular, \cite{ansari_within_2012} used the Beddington-DeAngelis incidence rate to model interaction between susceptible cells and free viruses. They argued that they could thus reduce the time needed for the immune response to clean the dengue virus, compared with using the mass action law of the paper \cite{nuraini_within_2009}. In both articles, depending on parameter values, more than one endemic steady-state was found, and thresholds for the stability of each of them have been attained. In \cite{clapham_within_2014}, the authors used data coming from hospitalized primary and secondary dengue patients - virus viremia - to estimate model parameters that confirmed the role of the immune response in shaping variation between individuals consistent with the hypothesis of ADE. In \cite{Anam2024}, the authors developed deterministic and stochastic within-host models to explore different dengue infection scenarios, taking into account individual immunological variability. Their models are calibrated using empirical data on viral load and antibody concentrations (IgM and IgG), incorporating confidence intervals derived from stochastic realizations. In \cite{gujarati_virus_2014}, the authors considered heterologous antibodies that interact through a Heaviside step function that triggers neutralization or enhancement of the infection depending on the amount of the antibody from the first infection. The model also considered a delay in the immune response which switches the stability of the system through a Hopf bifurcation. The authors argued that they could qualitatively replicate the humoral immune responses observed in primary and secondary infections. The authors of \cite{ben-shachar_minimal_2015} have shown that the risk of developing a severe form of dengue may be related to increased cytokine production due to the interaction of the immune system (T-cells) with the dengue virus. The model was parameterized to reproduce qualitatively the data set described in \cite{clapham_within_2014}. The virological indicators used were the level of peak viremia, the time to peak viremia, and the viral clearance rate. Although the model did not consider the humoral immune response explicitly, the authors argue that the reparametrization of the infectivity rate can be biologically interpreted in the context of the ADE phenomenon.

Other original approaches have been developed in \cite{adimy_maternal_2020, camargo_modeling_2021,ceron_simple_2018}. All, explicitly considered interactions among susceptible and infected target cells, dengue virus, and dengue antibodies through bilinear and trilinear terms. In \cite{ceron_simple_2018}, two thresholds were obtained: the mean number of virions produced by one invading virus in the very early stage of secondary infection, and the rate at which one infected macrophage dies. The existence and stability of the endemic equilibrium depend on both parameters, and one of them acts as a weakening factor for ADE. In \cite{Rubio2022}, the authors proposed a within-host mathematical model to evaluate the role of memory B and T cells in heterologous secondary dengue infection. They showed that memory T cells play an essential role in eliminating the possibility of ADE occurrence. In \cite{tang_modelling_2020}, the authors developed a within-host model to study the impact of ADE phenomena on disease severity of Zika virus and dengue virus sequential or co-infection. In \cite{Xu_Math_2024}, the authors developed a within-host model for primary and secondary dengue infections, taking into account the IgM and IgG antibody response. The model suggests that a faster rate of clearance of antibody-virus complexes may lead to a higher peak viral load and may explain the ADE phenomena in heterogeneous dengue infections.  In \cite{adimy_maternal_2020}, the authors explored the occurrence of Dengue Hemorrhagic Fever (DHF) in infants born from dengue-immune mothers, during their first dengue infection. The neutralizing and enhancing activities of maternal antibodies against the virus are represented by a function derived from experimental data. They were able to fit the model to data on the amount of maternal antibodies and the age of the infant at which DHF was reported. The authors were thus able to reproduce the time delay observed between the end of the protective level of maternal antibodies and the onset of hemorrhagic fever, in agreement with data from the literature. The article \cite{camargo_modeling_2021} focused on the ADE hypothesis and developed a mathematical model of secondary dengue infection by a different viral serotype. It considered indirect competition between neutralizing and enhancing antibodies. Here, the functions of enhancement and neutralization - which depend on the amount and type of antibody - are derived from basic concepts of chemical reactions and used to model virus-antibody complexes binding formed by distinct populations of antibodies, classified as cross-reactive or type-specific ones. The authors of this article concluded that virus-antibody immune complexes may promote viral clearance or enhancement of infection depending on the amount of cross-reacting antibodies and the rapid activation of the neutralizing antibodies. 

The authors of \cite{ perelson_modelling_2002, wodarz_mathematical_2002} considered the modeling of within-host dynamics of HIV infection and therapy. Both are review articles and can give a broad understanding of how the immune system interacts with viruses and the state of the art of HIV modeling which includes virus evolution. The authors claimed that HIV viruses have several different epitopes that can be recognized by immune response. Furthermore, different populations of HIV viruses and antibodies when mixed, may have cross-reactive responses. In \cite{pawelek_modeling_2012}, the proposed model includes both innate and adaptive immune responses, during an influenza virus infection. Different from the other models, a class of uninfected cells that are refractory to infections was included. Modeling predictions were compared with both interferon and viral kinetic data. The first post-peak viral decline is explained by the lysis of infected cells during the innate immune response, the subsequent viral plateau/second peak is generated by the loss of the IFN-induced antiviral effect and the increased availability of target cells. 

None of the cited works either considered the direct competition of the antibodies for the virus epitope or the synergistic interaction between antibodies. Therefore, the work presented here aims to introduce a new formalism, inspired by \cite{einav_when_2020,sanna_synergistic_2000}, to efficiently describe the formation of antigen-antibody complexes, and to use it in a mathematical model to describe the interaction between a virus, target cells and two antibodies competing to bind to virus receptors, taking into account their synergistic interaction, and the effectiveness of the complex formed in neutralizing or enhancing the virus. To achieve this, we first consider a monoclonal antibody and define the viral activity function of the antigen-antibody complex. We then consider a mixture of two antibodies and generalize the viral activity function obtained previously, which now depends on the amount of the two antibodies and the interaction between them. Finally, we use this viral activity function in a within-host infectious disease model. We derive some basic properties of the within-host model obtained, in particular the existence of steady-states (disease-free and endemic), and investigate the local and global asymptotic stability of these steady-states. We complete our study with numerical simulations to highlight various scenarios. Although we have parameterized the model to study homologous and heterologous secondary infection with the dengue virus, it can be easily adapted to other viruses. It can also be extended to vaccination or antibody therapy, \cite{klein_antibodies_2013, scott_antibody_2012}. To validate the proposed formalism, we fitted the within-host infection model to the data from \cite{clapham_modelling_2016}. The data include virus and antibody titer measurements recorded sequentially during a secondary infection from dengue patients. To understand the mechanisms underlying the heterogeneity of viral kinetics in patients, we focus on estimating parameters linked to the force of infection. The results show that the IgG and IgM antibodies generated during the first and second infections, respectively, have low synergistic binding. This means that they are in strong competition to bind to the virus which contributes to the enhancement of the infection. In addition, the strength of the interaction between antibodies varies from one individual to another, which explains the differences in the temporal evolution of virus and antibody populations and, therefore, the spectrum of dengue infection that includes the phenomenon of antibody-dependent enhancement (ADE).

\section{Modeling the antigen-antibody complex}

\subsection{Monoclonal antibody binding to a receptor}

Consider a monoclonal antibody $A$ that binds to an antigen $R$ (viral receptor) through a paratope-epitope bond and, depending on its affinity, inhibits or enhances viral activity. An antibody can recognize a number $n$ of epitopes of the same antigen, \cite{parren_antiviral_2001}. Antigen-antibody binding occurs through a chemical reaction, \cite{einav_when_2020}, in which $n$ paratopes $N$ of $A$ bind to $n$ epitopes of $R$. According to the standard law of mass action, the process leading to this bond is given by the following sequential chemical reaction:
\end{multicols}
\begin{equation}\label{Chemical Reaction}
    R+nN \xrightleftharpoons[k^-]{nk^+}{R_{1N}}+(n-1)N\xrightleftharpoons[2k^-]{(n-1)k^+} \dots \xrightleftharpoons[ik^-]{(n-i+1)k^+}R_{iN}+(n-i)N\xrightleftharpoons[(i+1)k^-]{(n-i)k^+}\dots\xrightleftharpoons[nk^-]{k^+}R_{nN},
\end{equation}
\begin{multicols}{2}
\noindent where $R_{iN}$ stands for the complex formed by $i$ epitopes of $R$ bound by $i$ paratopes $N$ of $A$, $R=R_{0N}$ means free antigen, $k^+$ and $k^-$ denotes the forward and backward reaction rates, \cite{perelson_immunology_1997}. When the chemical reaction reaches equilibrium, we obtain the following equation
\[(n-i)k^{+}[R_{iN}][N]^{n-i} = (i+1)k^{-}[R_{(i+1)N}][N]^{n-i-1},\]
valid for $i \in \{0, 1,\dots,n-1\}$, where $[\cdot]$ is the concentration of the chemical species. After some algebraic manipulation, we get 
\begin{equation*}
    [R_{iN}]=\dfrac{n-i+1}{i}\dfrac{[N]}{K}[R_{(i-1)N}],  \quad \text{for } i\in \{1, \dots, n\}, 
\end{equation*}
where $K=k^-/k^+$ is the dissociation constant of the reaction.

Solving this recurrence equation yields to
\begin{equation*}
    [R_{iN}]=C^i_n\left([N]/K\right)^i[R],
\end{equation*}
with $C^i_n=\frac{n!}{i!(n-i)!}$, the binomial coefficient. Given the total antigen concentration 
\[[R_{\text{total}}]=\sum_{i=0}^n[R_{iN}]=\left(1+\dfrac{[N]}{K}\right)^n[R],\]
and knowing that the antibody concentration $A$ is proportional to the concentration of paratope $N$, $A=\theta [N]$, the fraction of occupied epitopes of the antigen is given by
\begin{equation}\label{Fract}
    \frac{[R_{iN}]}{[R_{\text{total}}]}=\dfrac{C^i_n(A/K_{D})^i}{(1+A/K_{D})^n},
\end{equation}
with $K_D=\theta K$, where $\theta$ is a constant. The antigen-antibody complex modifies viral activity. Let's define the relative activity of a virus belonging to this complex as $\xi_{i} \geq 0$, where $i$ represents the number of occupied epitopes of the antigen. The value $0<\xi_i <1$ means that the complex antigen-antibody partially inhibits virus activity, while $\xi_i>1$ means that it increases virus activity. The case $\xi_{i}=1$ means that the relative virus activity remains unchanged. In particular, $\xi_{0}=1$, because no antigen-antibody complex is formed. As an example, in \cite{Flamand1993}, it was estimated that below 130 IgG or 30 IgM bound per virions, infectivity is totally preserved. Furthermore, binding to defective entry-mediating proteins on virions would not be directly relevant for neutralization, \cite{klasse_neutralization_2014}. Therefore, the activity of the virus is obtained by adding the products of the fractions \eqref{Fract} by their associated relative activities $\xi_i$, 
\begin{eqnarray}\label{Activity1}\begin{array}{ll}
\text{Activity}:=G(A)&=\sum_{i=0}^n\xi_{i}\times \dfrac{C^i_n(A/K_{D})^i}{(1+A/K_{D})^n},\\
&=\dfrac{1+\sum_{i=1}^n\xi_{i}C^i_n(A/K_{D})^i}{1+\sum_{i=1}^n C^i_n(A/K_{D})^i}. 
\end{array}
\end{eqnarray}

\begin{Figure}
   \centering
    \includegraphics[width=1\columnwidth]{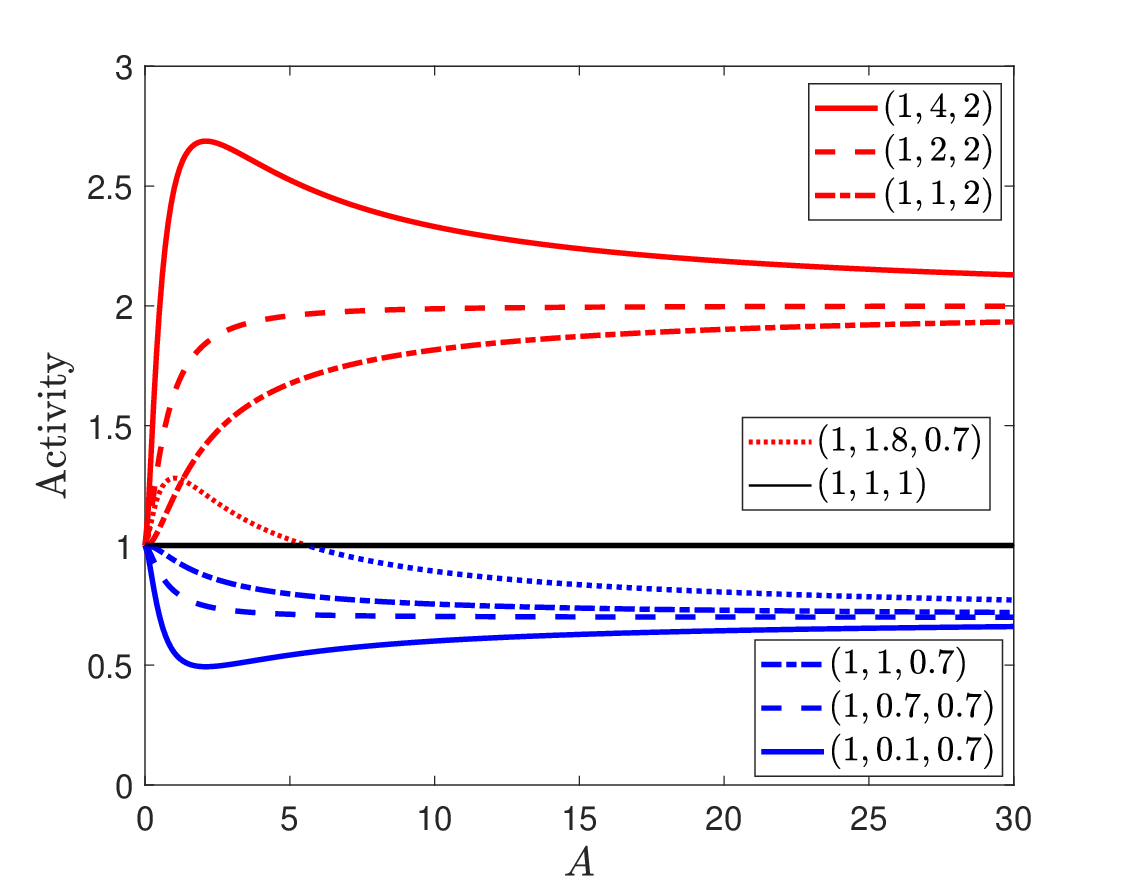}\
    \captionof{figure}{Viral activity as a function of monoclonal antibody $A$, G(A) (Equation \eqref{Activity1}). Each curve is associated with a specific set of relative viral activities $(\xi_1,\xi_2,\xi_3)$, $K_D=0.7$ and $n=3$. Red curves highlight the antigen-antibody complex leading to enhancement, while blue curves are those leading to neutralization. The black curve means that virus activity is neither neutralized nor enhanced. Depending on the amount of antibody, virus activity can change from enhancement to neutralization (red-blue dotted line).}
    \label{fig:act_one}
\end{Figure}

The binding affinity between the antigen's epitope and the antibody's paratope at a single binding site can be interpreted in different ways. For example, the higher the affinity, the lower the values of $K_D$ and $\xi_i$. This explains the property of neutralizing antibodies, which have a high affinity for the virus and, after binding to the receptor, neutralize it, $G(A)<1$ (blue lines in Figure \ref{fig:act_one}). Non-neutralizing cross-reactive antibodies, on the other hand, only partially recognize the receptor (low affinity), \cite{parren_antiviral_2001}. The constant $K_D$ is therefore high, and $\xi_i$ are close to $1$. In particular, if all relative activities $\xi_i$ are equal to $1$, the virus activity remains unaffected, $G(A)=1$ (black line in Figure \ref{fig:act_one}). For some viral infections, such as secondary dengue infection, non-neutralizing cross-reactive antibodies may enhance the infection. This means an increase in viral activity, $G(A)>1$ (red lines in Figure \ref{fig:act_one}). As the interaction antigen-antibody is a multi-hit phenomenon, \cite{pierson_structural_2008}, enhancement can occur for an intermediate amount of antibody, while neutralization occurs in the presence of a sufficiently large amount of antibody, \cite{durham_broadly_2019,katzelnick_antibody-dependent_2017}. This suggests that if the number of occupied epitopes is high enough, epitope-paratope binding will neutralize antigen activity; and if the number of occupied epitopes is intermediate, this binding will tend to enhance antigen activity (red-blue dotted line in Figure \ref{fig:act_one}).  

\subsection{Binding of a mixture of two antibodies to a receptor}\label{Subsection2.2}

In \cite{einav_when_2020}, the authors have developed a statistical mechanical model that predicts the collective efficacy of a mixture of antibodies whose constituents are assumed to bind to a single site on a receptor. We generalize their method to the case where the mixture can bind to multiple sites on a receptor, \cite{parren_antiviral_2001}. Let's consider now a mixture of two different antibodies $A_1$, $A_2$ that interact with a virus receptor $R$. This mixture can result in an antigen-antibody complex with purely independent bindings, purely competitive bindings, or synergistic bindings. This binding classification is based on the interactions (competition/cooperation) between the two antibodies. In this case, the activity of the virus is defined by the function $G(A_1, A_2)$, given by, 
\end{multicols}
\begin{equation}\label{General Activity}
    G(A_1,A_2)=\dfrac{1+\displaystyle\sum_{i=1}^{n_1}\xi_{1i}C^i_{n_1} \left(\frac{A_1}{K_{D_1}}\right)^i+\displaystyle\sum_{j=1}^{n_2}\xi_{2j}C^j_{n_2}\left(\frac{A_2}{K_{D_2}}\right)^j+\displaystyle\sum_{i=1}^{n_1}\displaystyle\sum_{j=1}^{n_2}f_{ij}\Tilde{\xi}_{1i}\Tilde{\xi}_{2j}C^i_{n_1}C^j_{n_2}\left(\frac{A_1}{\tilde{K}_{D_1}}\right)^i\left(\frac{A_2}{\tilde{K}_{D_2}}\right)^j}{1+\displaystyle\sum_{i=1}^{n_1}C^i_{n_1}\left(\dfrac{A_1}{K_{D_1}}\right)^i+\displaystyle\sum_{j=1}^{n_2}C^j_{n_2}\left(\frac{A_2}{K_{D_2}}\right)^j+\displaystyle\sum_{i=1}^{n_1}\displaystyle\sum_{j=1}^{n_2}f_{ij}C^i_{n_1}C^j_{n_2}\left(\frac{A_1}{\tilde{K}_{D_1}}\right)^i\left(\frac{A_2}{\tilde{K}_{D_2}}\right)^j}.
\end{equation}
\begin{multicols}{2}
As in \eqref{Activity1}, $n_1$ (resp. $n_2$) is the number of epitopes of the antigen that can be recognized by the antibody $A_1$ (resp. $A_2$), $\xi_{1i}$ (resp. $\xi_{2j}$) is the relative activity of the virus when bound to $i$ paratopes of $A_1$ (resp. $j$ paratopes of $A_2$), and $K_{D_1}$ (resp. $K_{D_2}$) is the dissociation constant associated to the antibody $A_1$ (resp. $A_2$). The new parameters $\Tilde{\xi}_{1i}$, $\Tilde{\xi}_{2j}$ are the relative virus activities modified by a synergistic binding (the relative virus activity may decrease in the presence of the other antibody because of competition, $\Tilde{\xi}_{kl} < \xi_{kl}$, or it can increase because of cooperation, $\Tilde{\xi}_{kl} > \xi_{kl}$). Similarly, a synergistic interaction between two antibodies can modify their binding to the receptor and thus their dissociation constants $K_{D_k}$, $k=1,2$. We then introduce modified dissociation constants $\tilde{K}_{D_k}$, with $\tilde{K}_{D_k}<K_{D_k}$ when antigen-antibody binding becomes stronger and $\tilde{K}_{D_k}>K_{D_k}$ when antigen-antibody binding becomes weaker. The coefficient $0\leq f_{ij}\leq 1$ corresponds to the fraction of simultaneous binding of both antibodies. Thus, three scenarios are possible: (i) purely competitive binding, with $f_{ij}=0$; (ii) synergistic binding, where $0< f_{ij}\leq 1$; and (iii) purely independent binding, where  $f_{ij}=1$, $\Tilde{\xi}_{ki}=\xi_{ki}$ and $\Tilde{K}_{D_k}=K_{D_k}$, for $k=1,2$, $i=1,\dots,n_1$ and $j=1,\dots,n_2$. In the case of \cite{einav_when_2020}, $n_1=n_2=1$. 
The purely competitive binding and the purely independent binding lead respectively to the following expressions of the virus activity
\end{multicols}

\begin{equation}\label{eq:G_competitive}
    G(A_1,A_2)=\dfrac{1+\displaystyle\sum_{i=1}^{n_1}\xi_{1i}C^i_{n_1} \left(\dfrac{A_1}{K_{D_1}}\right)^i+\displaystyle\sum_{j=1}^{n_2}\xi_{2j}C^j_{n_2}\left(\dfrac{A_2}{K_{D_2}}\right)^j}{1+\displaystyle\sum_{i=1}^{n_1}C^i_{n_1}\left(\dfrac{A_1}{K_{D_1}}\right)^i+\displaystyle\sum_{j=1}^{n_2}C^j_{n_2}\left(\dfrac{A_2}{K_{D_2}}\right)^j},
\end{equation}
and  
\begin{eqnarray}
\label{eq:G_independent} G(A_1,A_2)=\left(\dfrac{1+\displaystyle\sum_{i=1}^{n_1}\xi_{1i}C^i_{n_1}\left(\dfrac{A_1}{K_{D_1}}\right)^i}{\left(1+\dfrac{A_1}{K_{D_1}}\right)^{n_1}}\right)\times \left(\dfrac{1+\displaystyle\sum_{j=1}^{n_2}\xi_{2i}C^i_{n_2}\left(\dfrac{A_2}{K_{D_2}}\right)^j}{\left(1+\dfrac{A_2}{K_{D_2}}\right)^{n_2}}\right).
 \end{eqnarray}
To study the synergistic binding effect on viral activity, we introduce the following notations $\tilde{\nu}_{ij}=\Tilde{\xi}_{1i}\Tilde{\xi}_{2j}/(\xi_{1i}\xi_{2j})$ and $\tilde{\kappa}_{ij}=K_{D_1}^i K_{D_2}^j/(\tilde{K}_{D_1}^i\tilde{K}_{D_2}^j)$, and explore the values of these relative quantities. Therefore, the expression \eqref{General Activity} becomes
\begin{equation}\label{General ActivityBis}
    G(A_1,A_2)=\dfrac{1+\displaystyle\sum_{i=1}^{n_1}\xi_{1i}C^i_{n_1} \left(\frac{A_1}{K_{D_1}}\right)^i+\displaystyle\sum_{j=1}^{n_2}\xi_{2j}C^j_{n_2}\left(\frac{A_2}{K_{D_2}}\right)^j+\displaystyle\sum_{i=1}^{n_1}\displaystyle\sum_{j=1}^{n_2}\tilde{\nu}_{ij}\tilde{\kappa}_{ij}f_{ij}\xi_{1i}\xi_{2j}C^i_{n_1}C^j_{n_2}\left(\frac{A_1}{K_{D_1}}\right)^i\left(\frac{A_2}{K_{D_2}}\right)^j}{1+\displaystyle\sum_{i=1}^{n_1}C^i_{n_1}\left(\dfrac{A_1}{K_{D_1}}\right)^i+\displaystyle\sum_{j=1}^{n_2}C^j_{n_2}\left(\frac{A_2}{K_{D_2}}\right)^j+\displaystyle\sum_{i=1}^{n_1}\displaystyle\sum_{j=1}^{n_2}\tilde{\kappa}_{ij}f_{ij}C^i_{n_1}C^j_{n_2}\left(\frac{A_1}{K_{D_1}}\right)^i\left(\frac{A_2}{K_{D_2}}\right)^j}.
\end{equation}
\begin{multicols}{2}
The new parameters $\tilde{\nu}_{ij}$ and $\tilde{\kappa}_{ij}$ can be interpreted as follows. If $\tilde{\nu}_{ij}<1$, the relative viral activity due to synergistic binding decreases. On the other hand, if $\tilde{\nu}_{ij}>1$, relative viral activity increases. Furthermore, if $\tilde{\kappa}_{ij}<1$, the binding between the antibodies and the virus is enhanced by synergistic binding, and if $\tilde{\kappa}_{ij}<1$ this binding becomes weaker. In theory, many combinations are possible, but we will explore the most relevant ones: when both antibodies neutralize the virus activity, and when one antibody neutralizes and the other enhances the virus activity. 
\begin{figure*}
    \centering
    \subfigure[Neutralizing/Neutralizing]{\includegraphics[width=8cm,height=6cm]{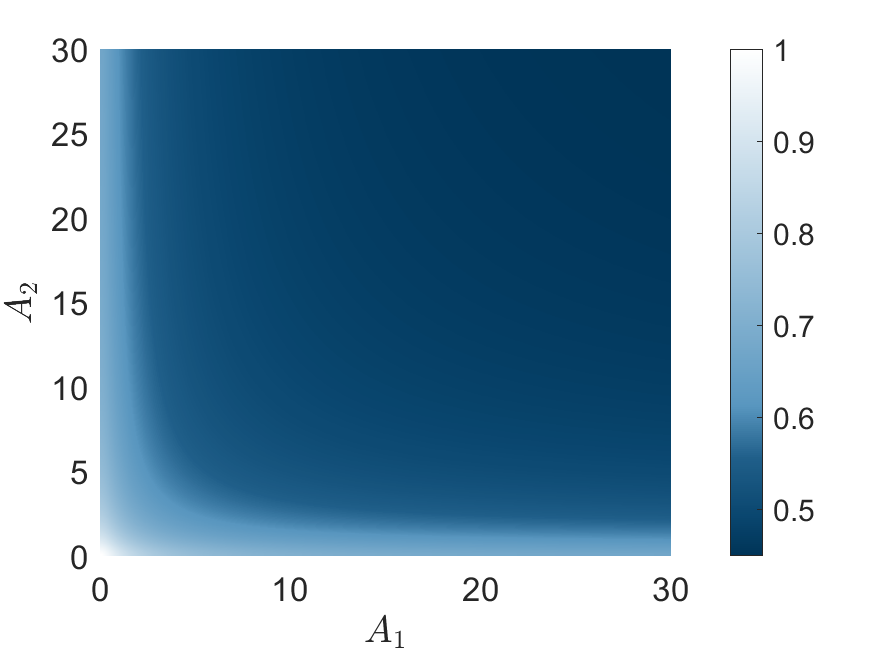}} 
\subfigure[Enhancing/Neutralizing]{\includegraphics[width=8cm,height=6cm]{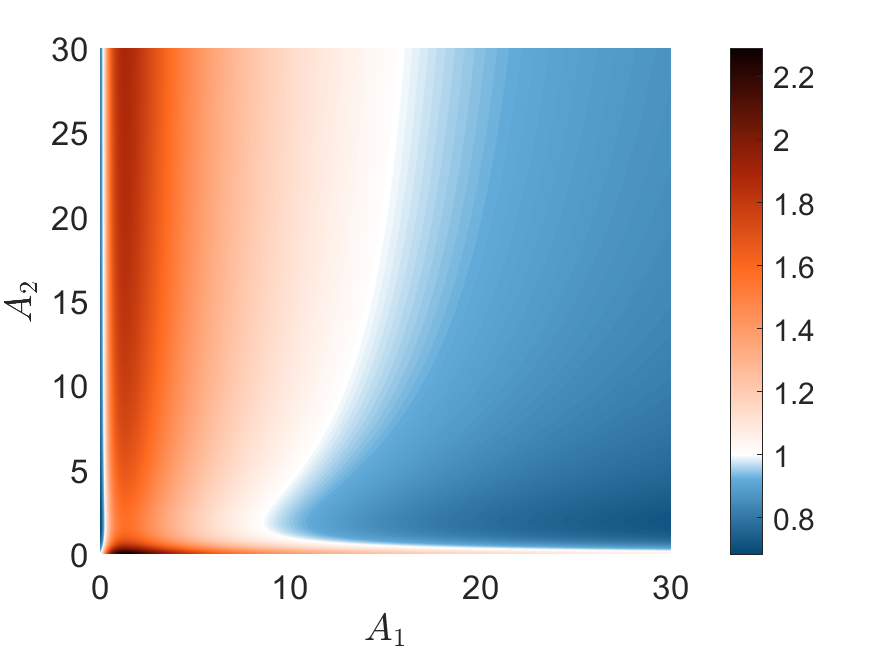}}
    \caption{Viral activity $G(A_1, A_2)$, in the case of Neutralizing/Neutralizing (left panel) and Enhancing/Neutralizing (right panel) antibodies. The case of
purely independent binding is explored (Equation \eqref{eq:G_independent}). The palette of colors, from blue (neutralization) to red (enhancement), reflects the variation in viral activity on antigen-antibody complexes.}
    \label{fig:Act3D}
\end{figure*}

Figure \ref{fig:Act3D} shows the behavior of viral activity as a function of the two antibodies $A_1$ and $A_2$. The case of purely independent binding is explored. In Figure \ref{fig:Act3D}(a), the two antibodies act to neutralize the virus. Whatever the direction, horizontal or vertical from the left or bottom border, we always move towards lower viral activity, which means more neutralization. This is coherent with the assertion of \cite{einav_when_2020,puschnik_correlation_2013} that two antibodies are more effective than one in neutralizing a virus. In the right panel of Figure \ref{fig:Act3D}, \textit{i.e.} Figure \ref{fig:Act3D}(b), it is shown the variation in viral activity in the case of enhancing antibody $A_1$ and neutralizing antibody $A_2$. For each fixed amount of neutralizing antibody, $A_2$, viral activity first increases very rapidly as the amount of enhancing antibody $A_1$ increases, reaching a peak in the red zone, then slowly decreases (moving from red to blue) until reaches a threshold in the blue (same behavior of the red/blue line of Figure \ref{fig:act_one}). On the other hand, if we fix the amount of enhancing antibody, $A_1$, and increase the amount of the neutralizing one, $A_2$, we obtain different scenarios: 
(i)	If $A_1$ is fixed at very low or high levels, viral activity will decrease rapidly at first, then reach a certain value (remain in the blue zone once entered, meaning that viral activity is neutralized);
(ii) If $A_1$ is fixed at an intermediate level, viral activity will start at a high level and change very quickly from red to blue or from dark red to light red, then return to an intermediate value and remain at this level. The main information to be deduced from Figure \ref{fig:Act3D}(b) is that infection is more severe when the amount of enhancing antibodies is in the intermediate range. This is compatible with the assertion of \cite{camargo_modeling_2021, durham_broadly_2019, katzelnick_antibody-dependent_2017} that the risk of severe dengue in a secondary heterologous infection is higher when there is an intermediate amount of pre-existing antibodies from the first infection. The common values of the parameters in Figure \ref{fig:Act3D} are $K_{D_1}=K_{D_2}=0.7$ mol ml$^{-1}$ and $\beta=8\times 10^{-9}$ ml RNA copies$^{-1}$. In Figure \ref{fig:Act3D}(a), we take $(\xi_{11},\xi_{12},\xi_{13})=(\xi_{21},\xi_{22},\xi_{23})=(1,0.95,0.65)$, and in Figure \ref{fig:Act3D}(b), we take $(\xi_{11},\xi_{12},\xi_{13})=(1,4,0.85)$ and $(\xi_{21},\xi_{22},\xi_{23})=(1,0.4,0.85)$.

The immune response to a viral infection can be extremely complex, orchestrated by a complex interaction between elements of the innate and adaptive immune response that depends on a number of factors such as host genetics, heterologous immunity, virus structure and receptor recognition, and avoidance of the immune response by the virus. Mechanisms preventing virion attachment to target cell receptors include binding at or near the viral receptor binding site and prevention of attachment by steric obstruction, disassembly or conformational modification of viral surface entry proteins, and virion aggregation.
In addition, the involvement of antibodies in virus dynamics also includes post-attachment neutralization, \cite{burton_antiviral_2023}. Nevertheless, in this article we try to describe the interactions between antibodies in a simple way. For example, we assume that viral surface molecules are static, but it is increasingly accepted that these molecules undergo conformational relaxation over time that exposes new epitopes or increases the exposure of existing epitopes. Cellular factors limit or enhance neutralization and indicate whether combinations produce additive, synergistic or antagonistic net effects. The list of antibody-virus interactions includes incomplete neutralization, synergistic neutralization, additive neutralization and enhancement; it does not exclude the existence of other interactions, \cite{klasse_neutralization_2014}.

\section{Within-host (re)infection dynamics in the presence of a mixture of two antibodies}\label{section:model}

As an example, let us consider a viral (re)infection at the cellular and immune level with a mixture of two antibodies. We note by $A_1$ the concentration of pre-existing antibody due to the first infection and by $A_2$ the concentration of the new antibody generated by the secondary infection with a homologous or heterologous virus. The quantities $X$ and $Y$ are the susceptible and infected target cells respectively, and $V$ is the free virus. Target cells may be epithelial cells in the case of influenza, \cite{kuiken_pathology_2008} or macrophages and dendritic cells in the case of dengue, \cite{kyle_dengue_2007}. The model is described by the ordinary differential system given by
\begin{equation}
\left\{\begin{aligned}
		\dfrac{dA_{1}}{d t} & =p_1(Y) - \gamma_{A_{1}}A_{1}, \\
		\dfrac{dA_{2}}{d t} & = p_2(Y) - \gamma_{A_{2}}A_{2}, \\
		\dfrac{d X}{d t} & = \Omega - \gamma_{X}X - \beta G(A_1,A_2)VX, \\
		\frac{d Y}{d t} & = \beta G(A_1,A_2)VX - \gamma_{Y}Y, \\
		\dfrac{d V}{d t} & = \Phi Y - (\delta + \gamma_{V}) V.
	\end{aligned}\right.
	\label{Newmodel2}
\end{equation}

The antibodies $A_i$ can be produced by memory B-cells or by plasmablasts upon stimulation by the virus or by infected cells, \cite{dorner_antibodies_2007}. In the case of a secondary viral infection, the first antibody $A_1$ is generated mainly by memory B-cells, while the second antibody $A_2$ is generated mainly by plasmablasts. It is reasonable to assume that the production of both antibodies depends on the concentration of infected cells. The production of antibody $A_i$ is then considered with a rate equal to $p_i(Y)$, $p_i$ being a nonnegative, continuously differentiable, and nondecreasing function on $[0,+\infty)$. For all numerical simulations, we use the function $p_i(Y)=\Lambda_i+\alpha_iY$. The natural mortality rates are denoted by $\gamma_{A_{i}}$, $\gamma_X$, $\gamma_Y$ and $\gamma_V$. The parameter $\Omega$ is the rate of production of susceptible target cells, which are generally produced in the bone marrow, \cite{trouplin_bone_2013}. $\Phi$ is the rate of virus production by infected cells,  \cite{rodenhuis-zybert_dengue_2010,samji_influenza_2009}. The parameter $\delta$ represents the rate of free virus loss by other means than natural mortality, such as neutralization or entry into target cells. When the two antibodies $A_1$ and $A_2$ interact with the virus, its activity may decrease or increase depending on the nature of these antibodies. This activity has been defined in Subsection \ref{Subsection2.2} by the function $G(A_1, A_2)$ given by \eqref{General ActivityBis}. Then, the force of infection is $\beta G(A_1,A_2)V$, with $\beta>0$. This system can model a homologous or heterologous secondary viral infection, such as dengue fever or influenza. It can also be adapted to antiretroviral therapy in the case of HIV, provided we change the functions $p_i$ and adapt the parameter values.  

Obviously, a simple mathematical model cannot fully reflect the real behavior of all the examples cited above. However, we are interested in an important part of the immune response, common to all these diseases, which concerns the mixture of competing or cooperating antibodies during infection.

Throughout this paper, we will need to make certain assumptions.
\begin{enumerate}[label=\textbf{(H\arabic*)}, ref=H\arabic*]
	\item \label{Hyp:gamma_X} $\gamma:=\gamma_X=\gamma_Y$.
	\item \label{Hyp:gamma_A} $\gamma_A:=\gamma_{A_1}=\gamma_{A_2}$.
	\item \label{Hyp:Reduction} $ G\left(A_1,A_2\right)\leq G\left(\frac{p_{1}(0)}{\gamma_{A_1}},\frac{p_{2}(0)}{\gamma_{A_2}}\right), \quad \text{for all } \left(A_1,A_2\right) \in \mathbb{D},$
with\[\mathbb{D}:=\left\{(A_1,A_2)\in\mathbb{R}^2_+:\; A_1\geq \dfrac{p_{1}(0)}{\gamma_{A_1}},\;A_2\geq \dfrac{p_{2}(0)}{\gamma_{A_2}}\right\}.\]
	\item\label{Hyp:Condition_f} $  p_{i}(Y)-p_{i}(0)\leq p_{i}'(0)Y$, for $i=1,2$, and all $Y \geq 0.$
 \item \label{Hyp:Condition_f'(0)} $p_{1}'(0)>0$ or $p_{2}'(0)>0 .$
\end{enumerate}
The assumptions \eqref{Hyp:gamma_X} and \eqref{Hyp:gamma_A}, mean that the disease does not affect cell mortality and that both antibodies have the same mortality rate. They are used, for the sake of simplicity, in the proof of the global asymptotic stability of $P_0$ and/or in the proof of the local asymptotic stability of $P^{\star}$, the disease-free equilibrium and the endemic one, respectively. Under \eqref{Hyp:Reduction}, quantities of antibodies $A_1$ and $A_2$ higher than those of the disease-free equilibrium reduce viral activity. This gives an advantage to neutralization. Hypothesis \eqref{Hyp:Condition_f} means that the dynamic of antibody production follows, at best, a linear growth. If \eqref{Hyp:Condition_f'(0)} is satisfied, then $p_1$ or $p_2$ is strictly increasing in a neighborhood of $0$, which implies a production strictly positive of $A_1$ or $A_2$ around $0$. The hypotheses \eqref{Hyp:Reduction}, \eqref{Hyp:Condition_f} and \eqref{Hyp:Condition_f'(0)} are assumed satisfied in Theorem \ref{T2} and Theorem \ref{Thm:local_stab_P*}.

\section{Mathematical analysis of the within-host (re)infection model}

In this section, we analyze the ordinary differential system \eqref{Newmodel2}. In particular, we establish the existence, uniqueness, and positivity of solutions, determine the existence of steady-states, calculate the basic reproduction number, and investigate the local and global asymptotic stability of the disease-free steady-state and the local asymptotic stability of the endemic equilibrium. For convenience, all the proofs of the theoretical results (Theorems, Propositions and Lemma) are given in the appendix, Section \ref{Appendix1}.

\subsection{Basic properties, steady-states, and basic reproduction number}

\begin{proposition}\label{P1}
The solution of the initial value problem \eqref{Newmodel2}, associated with a nonnegative initial condition, is unique, nonnegative, and bounded on $[0,+\infty)$.
\end{proposition}
\begin{proof}
The proof is given in the appendix (Subsection \ref{Appendix1}\ref{sec:ProofProp}).
\end{proof}
Let $P = \left(A_{1}^{\star}, A_{2}^{\star}, X^{\star}, Y^{\star}, V^{\star}\right)$ be an equilibrium:
	\begin{equation}
		\left\{\begin{array}{ll}
		p_{1}(Y^{\star}) - \gamma_{A_1}A_{1}^{\star} &= 0,\\
	    p_{2}(Y^{\star}) - \gamma_{A_2}A_{2}^{\star} &= 0,\\
		\Omega - \gamma_{X}X^{\star} - \beta G(A_{1}^{\star},A_{2}^{\star})V^{\star}X^{\star}&= 0,\\
		\beta G(A_{1}^{\star},A_{2}^{\star})V^{\star}X^{\star} - \gamma_{Y}Y^{\star} &= 0,\\
		\Phi Y^{\star} -\left(\gamma_{V}+\delta\right) V^{\star} &= 0.
		\end{array}\right.
		\label{homo}
	\end{equation}
By solving \eqref{homo}, we obtain: the disease-free equilibrium $P_0$ and endemic equilibrium points $P^{\star}$. The disease-free equilibrium is given by 
\begin{equation}
    P_0=\left(\dfrac{p_{1}(0)}{\gamma_{A_1}},\dfrac{p_{2}(0)}{\gamma_{A_2}},\dfrac{\Omega}{\gamma_X},0,0\right).
    \label{P0}
\end{equation}

The next-generation matrix, \cite{van_den_driessche_reproduction_2002}, is used to define the bifurcation parameter $\mathcal{R}_0$.
For this, we consider the two dimensions infected subsystem - $(Y,V)^T$ in \eqref{Newmodel2} -  that describe
the production of new infections and changes in the state of the infected
individuals
\begin{equation*}
	\left\{\begin{aligned}
		\frac{d Y}{d t} & = \beta G(A_1,A_2)VX - \gamma_{Y}Y, \\
		\dfrac{d V}{d t} & = \Phi Y - (\delta + \gamma_{V}) V.
	\end{aligned}\right.
\end{equation*}
The rates of appearance of the newly infected individuals are given by 
\begin{equation*}
    W = \left(\begin{array}{cc}
        0 &  \beta \,G\left(\frac{p_{1}(0)}{\gamma_{A_{1}}}\,,\,\frac{p_{2}(0)}{\gamma_{A_{2}}}\right)\frac{\Omega}{\gamma_X} \\
        0 & 0 
    \end{array}\right),
\end{equation*}
and the rates of transfer of individuals within the infected compartment, by any means other than the appearance of newly infected individuals from the uninfected compartment, \cite{van_den_driessche_reproduction_2002}, are given by 
\begin{equation*}
    T = \left(\begin{array}{cc} -\gamma_Y & 0 \nonumber \\
        \Phi & -\left(\gamma_{V}+\delta\right)
    \end{array}\right).
\end{equation*} 
These two matrices are obtained by decomposing the Jacobian matrix into two parts, the transmission $W$ and the transition $T$, which are evaluated at the disease-free equilibrium $P_0:=\left(p_{1}(0)/\gamma_{A_{1}},p_{2}(0)/\gamma_{A_{2}},\Omega/\gamma_X,0,0\right)$. Finally, the bifurcation parameter $\mathcal{R}_0$ is the spectral radius of the matrix product $-WT^{-1}$, \cite{van_den_driessche_reproduction_2002},
\begin{equation*}
    \mathcal{R}_0=\rho \left(-WT^{-1}\right)=\dfrac{\Phi \beta \Omega}{\gamma_{Y}\gamma_{X}\left(\gamma_V+\delta\right)}G\left(\dfrac{p_{1}(0)}{\gamma_{A_{1}}},\dfrac{p_{2}(0)}{\gamma_{A_{2}}}\right).
\end{equation*}
$\mathcal{R}_0$ measures the expected number of infected cells (or viral particles) generated by an infected cell (or viral particle) in a cell population where all cells are assumed to be susceptible. As in epidemiology, it can be proven that if $\mathcal{R}_0<1$ the disease-free equilibrium $P_0$ is locally asymptotically stable, this means that the infection is acute, otherwise if $\mathcal{R}_0>1$, $P_0$ becomes unstable and the infection is chronic. 

\subsection{Local and global asymptotic stability of the disease-free equilibrium}

For local asymptotic stability and instability of the disease-free equilibrium $P_0$, only the conditions $\mathcal{R}_0<1$ and $\mathcal{R}_0>1$ are required. 
\begin{theorem}\label{T1}
The disease-free equilibrium point $P_0$ is locally asymptotically stable if $\mathcal{R}_0 <1$ and it is unstable if $\mathcal{R}_0> 1$.
\end{theorem}
\begin{proof}
The proof is given in the appendix (Subsection \ref{Appendix1}\ref{ProofsTh}).
\end{proof}

To analyze the global asymptotic stability of the disease-free steady-state $P_0$, we introduce the following subset $\mathbb{D}$ of $\mathbb{R}^2_+$,
\[\mathbb{D}:=\left\{(A_1,A_2)\in\mathbb{R}^2_+: A_1\geq \dfrac{p_{1}(0)}{\gamma_{A_1}},\;A_2\geq \dfrac{p_{2}(0)}{\gamma_{A_2}}\right\}.\]
To prove the global asymptotic stability of $P_0$, we have to check two cases, whether the system starts with a nonnegative initial condition inside or outside the set $\mathbb{D} \times \mathbb{R}^3_+$. Before that, we need the following invariance result.
\begin{lemma}\label{Lem:D*Rinvariant}
Assume \eqref{Hyp:gamma_X}. Then, the subset $\mathbb{D}\times \mathbb{R}^3_+$ is invariant under System \eqref{Newmodel2}.
\end{lemma}
\begin{proof}
The proof is given in the appendix (Subsection \ref{Appendix1}\ref{sec:ProofLemma}).
\end{proof}

We use a comparison result, Lemma 1 of \cite{camargo_modeling_2021} and \cite{Kirkilionis_comparison_2004}, to prove the global asymptotic stability of $P_0$.
\begin{theorem}\label{T2}
Assume \eqref{Hyp:gamma_X}, \eqref{Hyp:Reduction},\eqref{Hyp:Condition_f}, \eqref{Hyp:Condition_f'(0)} and $\mathcal{R}_0<1$. Then, the disease-free equilibrium $P_0$ is globally asymptotically stable.
\end{theorem}
\begin{proof}
The proof is given in the appendix (Subsection \ref{Appendix1}\ref{ProofsTh}).
\end{proof}
\begin{remark}\label{R1}
We give here an example of a function $G$ that satisfies Condition \eqref{Hyp:Reduction}. We consider the case of a purely independent binding \eqref{eq:G_independent} and we assume that one of the following points is satisfied.
\begin{itemize}
    \item $\xi_{1,i} \leq \xi_{1,i-1}\leq 1$ and $\xi_{2,j} \leq \xi_{2,j-1} \leq 1$, for all $2\leq i\leq n_1$, $2\leq j\leq n_2$, 
    \item $\xi_{1,n_1} \leq \xi_{1,n_1-1}$, $\xi_{2,n_2} \leq \xi_{2,n_2-1}$, and all others equal to $1$,
    \[\dfrac{p_1(0)}{\gamma_{A_1}}\geq (n_1-1)\dfrac{\xi_{1,n_1-1}-1}{\xi_{1,n_1-1}-\xi_{1,n_1}}\] \text{and} \[\dfrac{p_2(0)}{\gamma_{A_2}}\geq (n_2-1)\dfrac{\xi_{2,n_2-1}-1}{\xi_{2,n_2-1}-\xi_{2,n_2}}.\]
\end{itemize}
Then, Condition \eqref{Hyp:Reduction} is satisfied.
\end{remark}
Indeed, in the case of a purely independent binding, we have $G\left(A_1,A_2\right)=G\left(A_1,0\right)\times G\left(0, A_2\right)$ and if one of the two points of Remark \ref{R1} is satisfied, then the two functions $A_1 \mapsto G\left(A_1,0\right)$ and $A_2 \mapsto G\left(0, A_2\right)$ are decreasing on $\mathbb{D}$. This means that Condition \eqref{Hyp:Reduction} is satisfied on $\mathbb{D}$. Let's give a biological interpretation of the two conditions in Remark \ref{R1}. Suppose that $n_1=n_2=3$. We only need to explain for $G\left(A_1,0\right)$ as the other is similar. The first condition becomes $\xi_{1,3} \leq \xi_{1,2} \leq \xi_{1,1} \leq 1$. This means that the increase in the number of epitopes recognized and bound by antibodies promotes the decrease of the virus activity. In any case, the formed immune complexes result in the neutralization of virus activity.  In this case, the function $A_1 \mapsto G(A_1,0)$ is always decreasing. The second condition becomes $\xi_{1,3} \leq \xi_{1,2}$, $\xi_{1,1}=1$ and $\dfrac{p_1(0)}{\gamma_{A_1}}\geq 2\dfrac{\xi_{1,2}-1}{\xi_{1,2}-\xi_{1,3}}$. With this condition, even if the relative activities are not less than $1$, the associated viral activity $A_1 \mapsto G(A_1,0)$ is decreasing on the set $\mathbb{D}$ because the amount of antibodies before the virus invasion is big enough to stop the infection.

\subsection{Existence of endemic equilibrium points}

\begin{theorem}\label{Thm:exist_P*}
Assume that $\mathcal{R}_0>1$. Then, there exists (at least) one endemic steady-state and all endemic steady-states belong to the subset $\mathbb{D}\times \mathbb{R}^3_+$.
\end{theorem}
\begin{proof}
The proof is given in the appendix (Subsection \ref{Appendix1}\ref{ProofsTh}).
\end{proof}
As mentioned before, Condition \eqref{Hyp:Reduction} promotes neutralization and it is not surprising in this case to have no endemic equilibrium if $\mathcal{R}_0<1$. We prove this in the following result.
\begin{proposition}\label{Prop:noEndemic}
  Assume that $\mathcal{R}_0<1$ and the function $G$ satisfies Condition \eqref{Hyp:Reduction} on the subset $\mathbb{D}$. Then, there is no endemic steady-state. 
\end{proposition}
\begin{proof}
The proof is given in the appendix (Subsection \ref{Appendix1}\ref{sec:ProofProp}).
\end{proof}
\begin{proposition}\label{P2}
Suppose that the function $Y^{\star}\mapsto G\left(p_{1}(Y^{\star})/\gamma_{A_1},p_{2}(Y^{\star})/\gamma_{A_2}\right)$ is decreasing on the interval $\left[0,\Omega/\gamma_{Y} \right]$. Then, if $\mathcal{R}_0< 1$ there is no endemic steady-state, otherwise if $\mathcal{R}_0> 1$ there is a unique endemic steady-state.
\end{proposition}
\begin{proof}
The proof is given in the appendix (Subsection \ref{Appendix1}\ref{sec:ProofProp}).
\end{proof}

\begin{remark}\label{R2}
Suppose that the function $Y^{\star}\mapsto G\left(p_{1}(Y^{\star})/\gamma_{A_1},p_{2}(Y^{\star})/\gamma_{A_2}\right)$ is not decreasing on the interval $\left[0,\Omega/\gamma_{Y} \right] $. Then, we may have the existence of: 
\begin{itemize}
    \item an endemic equilibrium, even in the case where $\mathcal{R}_0 < 1$, 
    \item more than one endemic equilibrium in the case where $\mathcal{R}_0 > 1$. 
\end{itemize}
\end{remark}

\subsection{Local asymptotic stability analysis of an endemic equilibrium point}

\begin{theorem}\label{Thm:local_stab_P*}
Assume that \eqref{Hyp:gamma_X}, \eqref{Hyp:gamma_A}, \eqref{Hyp:Reduction}, \eqref{Hyp:Condition_f} are satisfied, $\mathcal{R}_0>1$, and the following hypothesis 
\begin{equation}
    p'_1(Y^{\star}) G_{A_1}(A_1^{\star},A_2^{\star})+p'_2(Y^{\star}) G_{A_2}(A_1^{\star},A_2^{\star}) \leq 0,
    \label{Condition Derivative G}
\end{equation}
where $P^{\star}=(A_{1}^{\star},A_{2}^{\star},X^{\star},Y^{\star},V^{\star})$ is an endemic steady-state, $G_{A_1}$ and $G_{A_2}$ are the partial derivatives of the function $G$ with respect to the first and second variables. Then, $P^{\star}=(A_{1}^{\star},A_{2}^{\star},X^{\star},Y^{\star},V^{\star})$ is locally asymptotically stable.
\end{theorem}
\begin{proof}
The proof is given in the appendix (Subsection \ref{Appendix1}\ref{ProofsTh}).
\end{proof} 
\begin{remark}
Condition \eqref{Condition Derivative G} can be interpreted geometrically by writing it as a scalar product $\nabla p \cdot \left(\nabla G\right)^T \leq 0$ at the endemic equilibrium, where $p$ is the vector $(p_1,p_2)^T$, $T$ is for transpose and $\nabla$ is the gradient. As $p'_i(Y^{\star})$ is nonnegative, the vector $ \nabla p$ is located in the first quadrant, $(+,+)$. Then, $\left(\nabla G\right)^T$ has to be at least in the second or fourth quadrant, $(-,+)$ or $(+,-)$. If $\left(\nabla G\right)^T$ is in the third quadrant, $(-,-)$, Condition \eqref{Condition Derivative G} is always satisfied. It should also be pointed out that Condition \eqref{Condition Derivative G} is only sufficient but not necessary for the result of Theorem \ref{Thm:local_stab_P*} to be valid.
\end{remark}

\section{Numerical simulations}\label{NumSim}

We will now numerically study the system \eqref{Newmodel2} in two different scenarios: (i)
a mixture of two neutralizing antibodies (neutralizing/neutralizing); (ii) a mixture of enhancing and neutralizing antibodies (enhancing/neutralizing). Our main goal is to investigate how competition/cooperation between the two antibodies $A_1$ and $A_2$ can affect disease progression during secondary infection. Competition or cooperation between the two antibodies is primarily determined by their ability to bind to the same receptor. This is indicated by the parameters $f_{ij}$. Once binding to the same receptor, the antibodies can interact with synergistic binding through the parameters $\tilde{\nu}_{ij}$ and $\tilde{\kappa}_{ij}$.

For the numerical simulations, we use the functions $p_i(Y)=\Lambda_i+\alpha_iY$. The baseline parameters are taken from the literature (see \cite{adimy_maternal_2020, camargo_modeling_2021}, for further details) and are summarized in Table \ref{tab:Parameters_value}. They are associated with dengue infection. The number of parameters $f_{ij}$ is $9$. In our analysis, we will consider the case where $f_{ij}$ are independent of $i,j$. The initial conditions are $(A_1(0),A_2(0),X(0),Y(0),V(0))=(0,0,10^7,0,1)$. The units for $A_1$, $A_2$ are mol ml$^{-1}$, for $X$, $Y$ cells ml$^{-1}$, and for $V$ RNA copies ml$^{-1}$. 

In our approach, we fix the two parameters $\tilde{\nu}_{ij}$ and $\tilde{\kappa}_{ij}$ at a value corresponding to each of the two scenarios, neutralizing/neutralizing and enhancing/neutralizing, and vary $0\leq f_{ij}\leq 1$. Similar results can be obtained if we fix $f_{ij}$ and either $\tilde{\kappa}_{ij}$ or $\tilde{\nu}_{ij}$, and vary the other. 

\subsection{Homologous sequential DENV infection (Neutralizing/Neutralizing)}

We consider the case where both antibodies $A_1$ and $A_2$ have the same characteristics to neutralize the virus, $K_{D_1}=K_{D_2}$ and $(\xi_{11},\xi_{12},\xi_{13})=(\xi_{21},\xi_{22},\xi_{23})$. We choose $K_{D_1}=0.6$ and $(\xi_{11},\xi_{12},\xi_{13})=(1,0.95,0.8)$. We assume that their neutralizing effects are reinforced by the same synergistic binding, $\tilde{K}_{D_1}=\tilde{K}_{D_2}$ and $(\tilde{\xi}_{11},\tilde{\xi}_{12},\tilde{\xi}_{13})=(\tilde{\xi}_{21},\tilde{\xi}_{22},\tilde{\xi}_{23})$. We take the values $\tilde{K}_{D_1}=0.8$ and $(\tilde{\xi}_{11},\tilde{\xi}_{12},\tilde{\xi}_{13})=(\sqrt{0.65},\,0.95 \times \sqrt{0.65},\,0.8\sqrt{0.65})$. This corresponds to $\tilde{\nu}_{ij}=0.65<1$ and $\tilde{\kappa}_{ij}=1$, for all $i,j=1,2,3$.

We first consider the case where there is only one neutralizing antibody $A_2$, taking in Figure \ref{fig:evo_rev1}(a), $\Lambda_1=0$ and $\alpha_1=0$. We then examine what happens when a second neutralizing antibody is introduced, taking in Figure \ref{fig:evo_rev1}, (b)-(d), $\Lambda_1=7.5\times 10^3$ and $\alpha_1=0.4$. Finally, we vary the parameter $f_{ij}$, which simulates the fraction of simultaneous binding of the two antibodies to the same receptor, in each sub-figure (b), (c), and (d) of Figure \ref{fig:evo_rev1}, by taking: (b) $f_{ij}=0$, (c) $f_{ij}=10^{-6}$ and (d) $f_{ij}=1$. Figures \ref{fig:evo_rev1}, (a)-(d), show that the presence of two neutralizing antibodies is at least more efficient than a single one. Figure \ref{fig:evo_rev1}(b) shows that purely competitive binding, $f_{ij}=0$, of two neutralizing antibodies, while causing a very rapid decrease in viral activity, does not significantly increase neutralization, compared with a single neutralizing antibody, Figures \ref{fig:evo_rev1}(a). As the synergistic binding is increased, neutralization is improved. Indeed, as the parameter $f_{ij}$ increases, the proportion of infected cells is reduced and the amount of virus at peak is lower, while infection also occurs later. 

\subsection{Heterologous sequential DENV infection (Enhancing/Neutralizing)}

We consider one enhancing antibody, $A_1$, and one neutralizing antibody, $A_2$. We choose $K_{D_1}=\tilde{K}_{D_1}=K_{D_2}=\tilde{K}_{D_2}=0.8$ and $(\xi_{11},\xi_{12},\xi_{13})=(\tilde{\xi}_{11},\tilde{\xi}_{12},\tilde{\xi}_{13})=(3,3,3)$. We assume that the relative viral activities associated with the second antibodies are such that $\xi_{1i} \xi_{2j}=\tilde{\xi}_{1i} \tilde{\xi}_{2j}=1$. With this choice, we have $\tilde{\kappa}_{ij}=1$ for all $i,j$. We also choose the parameter $\tilde{\nu}_{ij}$ with the same value for all $i,j$. Two situations can be considered, one with $\tilde{\nu}_{ij}<1$, meaning that the sum of the effects of the two antibodies leads to neutralization, and the other with $\tilde{\nu}_{ij}>1$, meaning that the sum of the effects of the two antibodies produces enhancement. We focus here on the case with $\tilde{\nu}_{ij}=0.4<1$.

Let's first consider the case of a single neutralizing antibody $A_2$, by taking in Figure \ref{fig:evo_rev2}(a), $\Lambda_1=0$ and $\alpha_1=0$. The parameters were chosen to obtain a high degree of neutralization. We then present the enhancing antibody $A_1$, by considering in Figure \ref{fig:evo_rev2}, (b)-(d), $\Lambda_1=7.5\times 10^3$ and $\alpha_1=0.4$. As in the previous case, we study the effect of synergistic binding by increasing the parameter $f_{ij}$. More precisely, for $f_{ij}=0$, a purely competitive binding, Figure \ref{fig:evo_rev2}(b) shows that the antibody $A_1$ strongly increases infection, with a peak of the virus reached very quickly and viral activity remaining in the red zone (refereed to Figure \ref{fig:Act3D}). As synergistic binding increases, Figure \ref{fig:evo_rev2}: (b) $f_{ij}=0$, (c) $f_{ij}=10^{-3}$, (d) $f_{ij}=1$, neutralization becomes more effective than enhancement. Indeed, the proportion of infected cells is reduced and the amount of virus at peak is lower, while infection also occurs later. In the same way, viral activity falls into the blue zone (referred to Figure \ref{fig:Act3D}). Based on these comparisons, we can affirm that purely competitive binding, $f_{ij}=0$, gives an advantage to the enhancing antibody while purely independent binding, $f_{ij}=1$, confers a greater competitive advantage to the neutralizing antibody. 


\begin{figure*}
   \centering
(a) Single neutralizing antibody ($\Lambda_1=\alpha_1=0$)\\ 
\subfigure{{\includegraphics[scale=0.34]{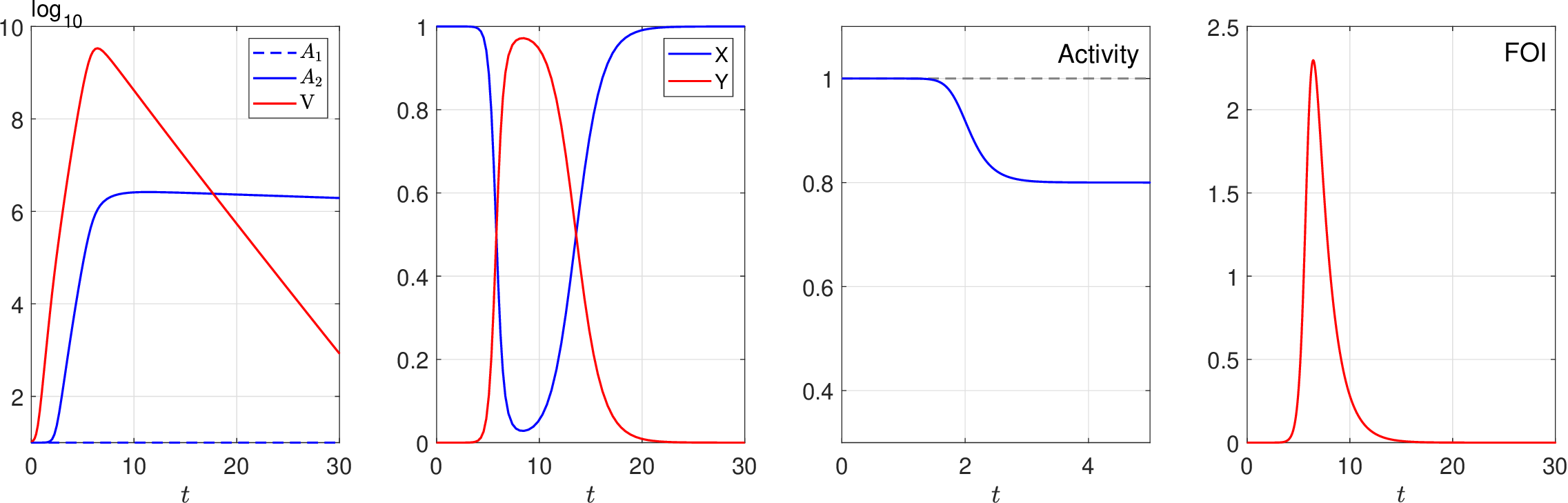}}
}\\
(b) Two homologous antibodies with purely competitive binding ($f_{ij}=0$) \\

\subfigure{{\includegraphics[scale=0.34]{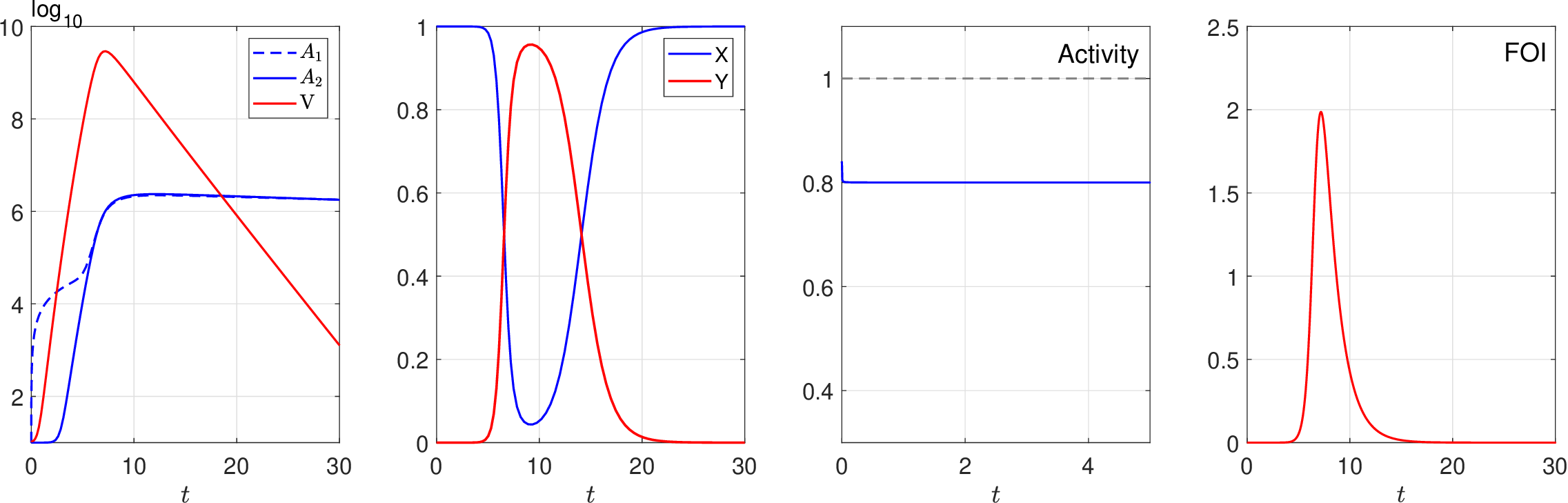}}
}\\
(c) Two homologous antibodies with low synergistic binding ($f_{ij}=10^{-6}$)\\
\subfigure{{\includegraphics[scale=0.34]{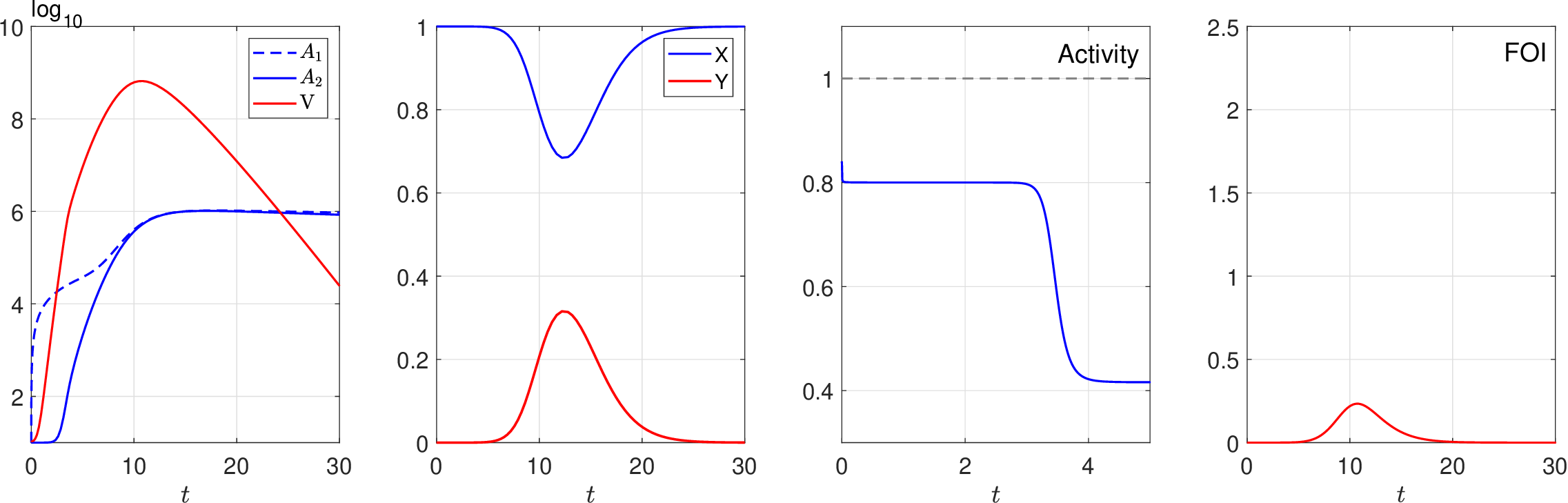}}
}\\
(d) Two homologous antibodies with high synergistic binding ($f_{ij}=1$)\\
\subfigure{{\includegraphics[scale=0.34]{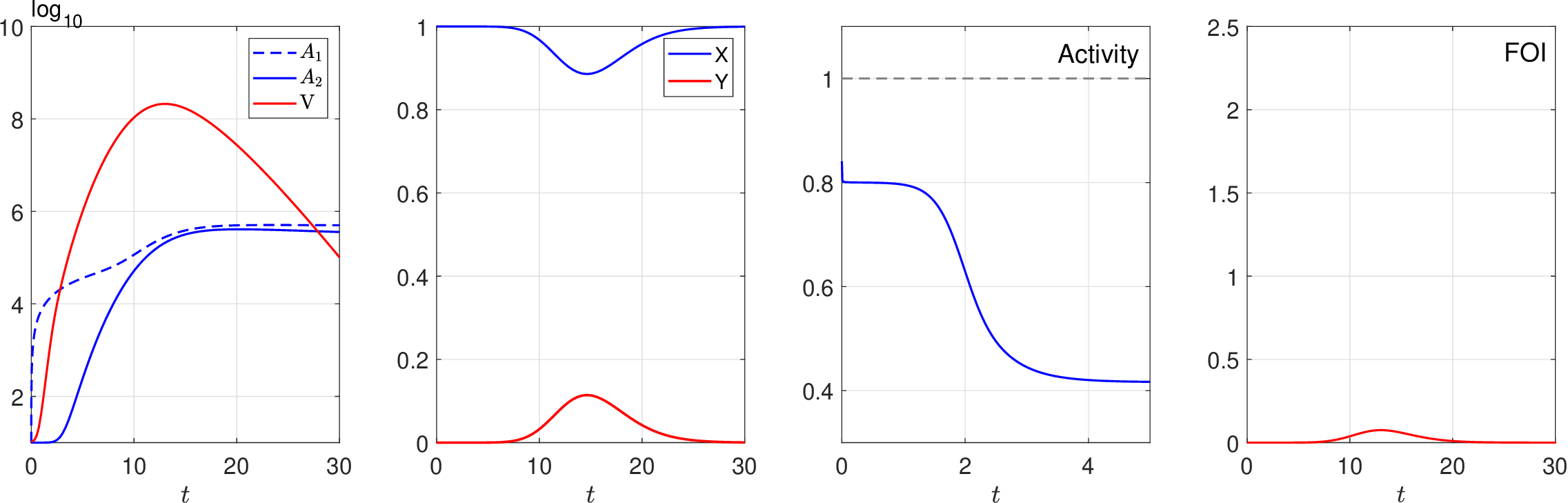}}
}
\caption{Neutralizing/Neutralizing. Here we compare the temporal dynamics of antibodies, cells, and viral populations during a primary viral infection and a secondary infection by a homologous virus, by varying the competition between the two antibodies to bind to the virus, (b) $f_{ij}=0$, (c) $f_{ij}=10^{-6}$ and (d) $f_{ij}=1$. From left to right, the first column of each panel shows the temporal evolution of the neutralizing antibody $A_1$ (dotted blue line), the neutralizing antibody $A_2$ (solid blue line), and the virus $V$ (solid red line). The second column shows the evolution of the proportion of susceptible and infected cells $X$ and $Y$ (solid blue and red lines). The third column shows the temporal activity of the virus (solid blue line). The gray dotted line indicates the threshold of $1$, i.e., when the complex antigen-antibody does not affect virus activity. The last column shows the force of infection, FOI (red line). The parameter values are given in Table \ref{tab:Parameters_value}.
}
\label{fig:evo_rev1}
\end{figure*}

\begin{figure*}

   \centering 
(a) Single neutralizing antibody ($\Lambda_1=\alpha_1=0$)\\ 
\subfigure{{\includegraphics[scale=0.34]{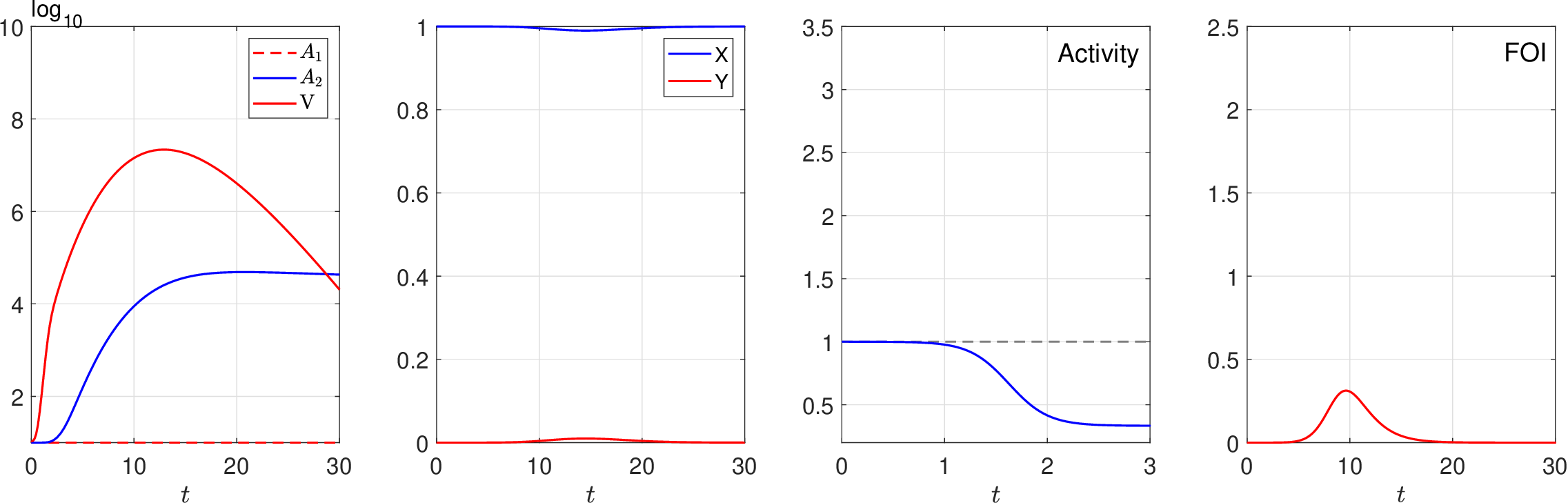}}
}\\
(b) Two heterologous antibodies with purely competitive binding ($f_{ij}=0$) \\
\subfigure{{\includegraphics[scale=0.34]{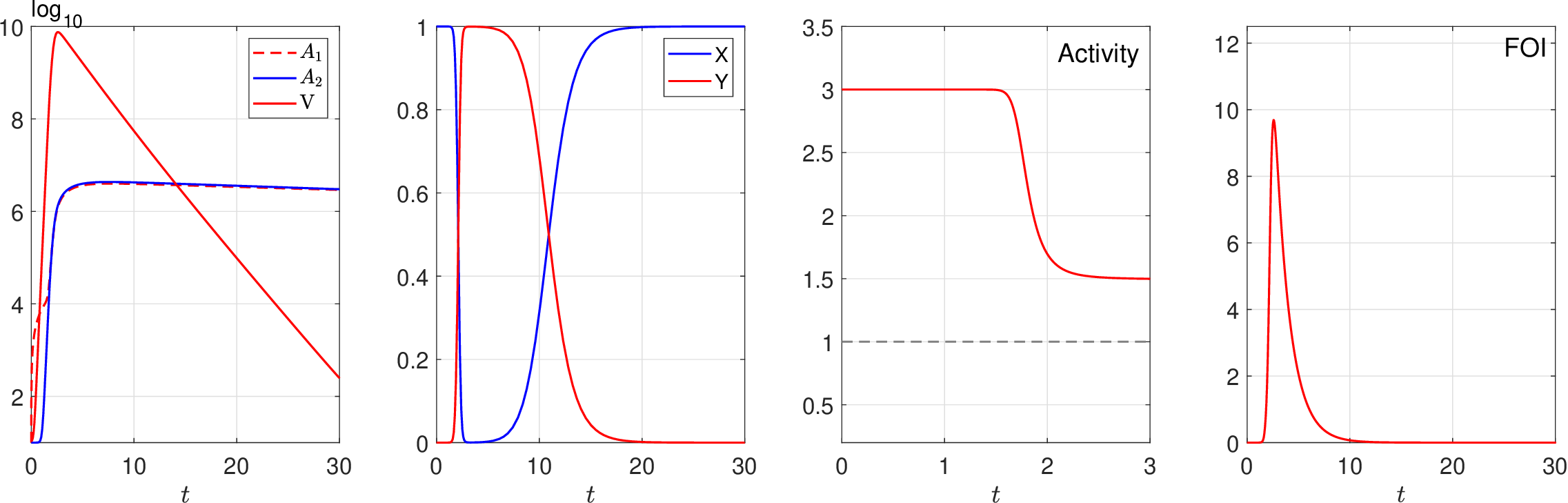}}
}\\
(c) Two heterologous antibodies with low synergistic binding ($f_{ij}=10^{-3}$)\\
\subfigure{{\includegraphics[scale=0.34]{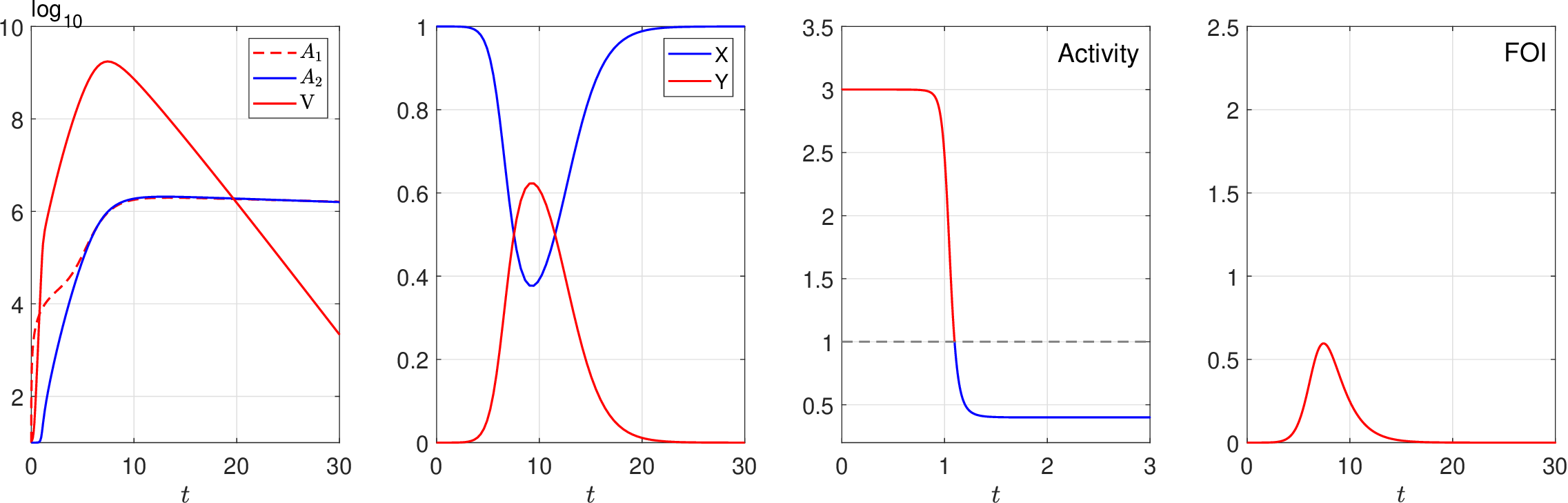}}
}\\
(d) Two heterologous antibodies with high synergistic binding ($f_{ij}=1$)\\
\subfigure{{\includegraphics[scale=0.34]{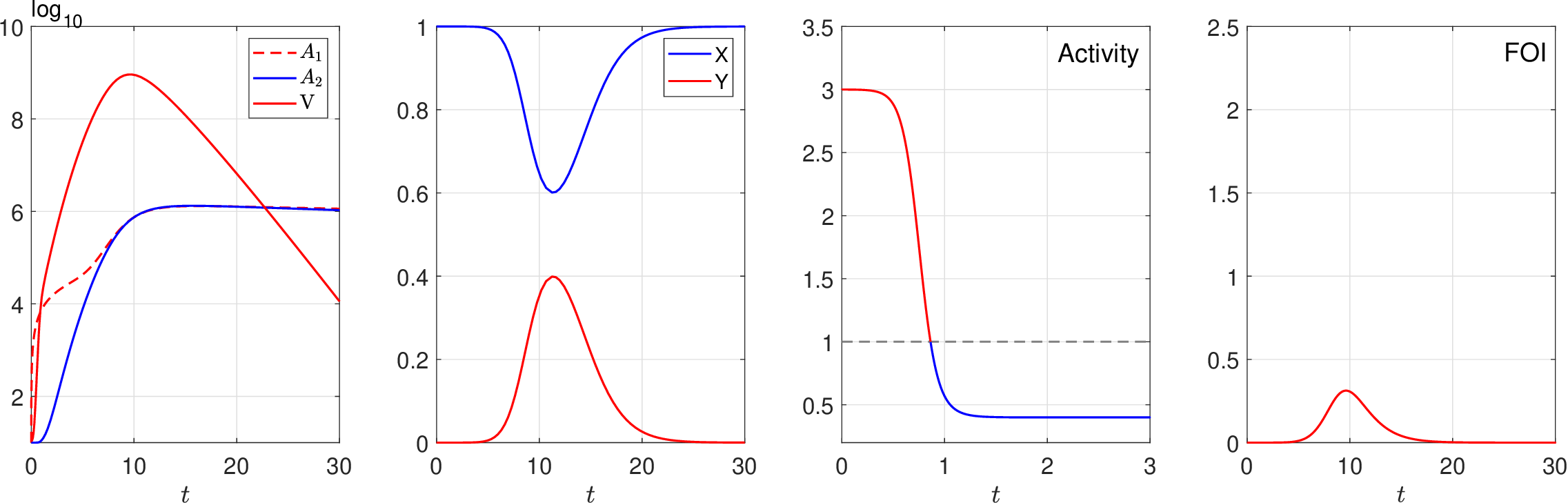}}
}
\caption{Enhancing/Neutralizing. Here we compare the temporal dynamics of antibodies, cells, and viral populations during a primary viral infection and a secondary infection by a heterologous virus, by varying the competition between the two antibodies to bind to the virus, (b) $f_{ij}=0$, (c) $f_{ij}=10^{-3}$ and (d) $f_{ij}=1$. From left to right, the first column of each panel shows the temporal evolution of the enhancing antibody $A_1$ (dotted red line), the neutralizing antibody $A_2$ (solid blue line), and the virus $V$ (solid red line). The second column shows the evolution of the proportion of susceptible and infected cells $X$ and $Y$ (solid blue and red lines). The gray dotted line indicates the threshold of $1$, i.e., when the complex antigen-antibody does not affect virus activity. The third column shows the temporal virus activity (solid blue-red line). The last column shows the force of infection, FOI (red line). The parameter values are given in Table \ref{tab:Parameters_value}.
}
\label{fig:evo_rev2}
\end{figure*}

\begin{figure*}
\centering
    \includegraphics[scale=0.55]{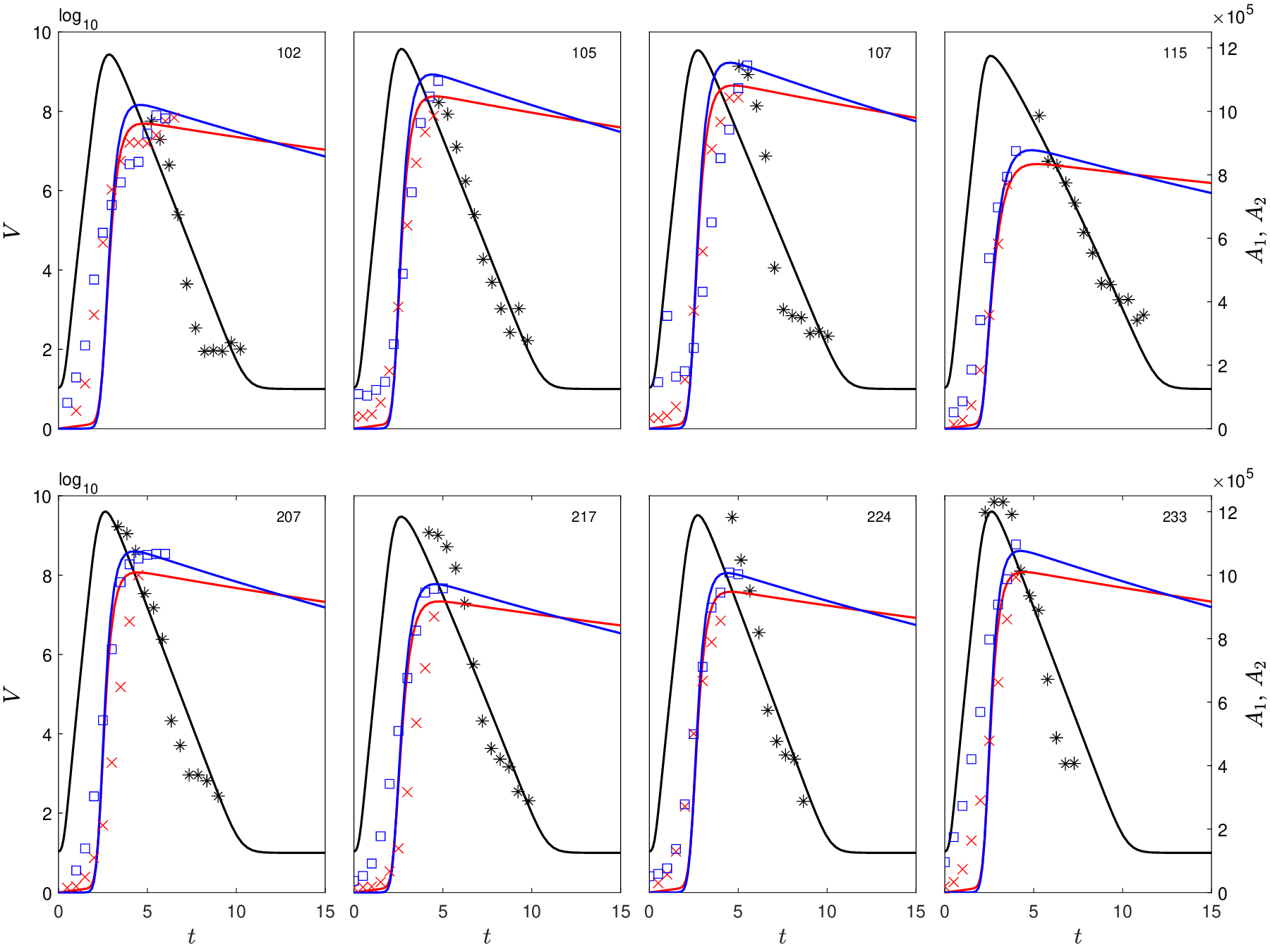}
        \caption{Temporal evolution of enhancing antibodies (red curves), neutralizing antibodies (blue curves), and viral populations (black curves) during secondary infection with a heterologous virus. Each plot corresponds to an individual infected with dengue virus serotype II. Patient data are represented by symbols: {{\color{red}$\times$} IgG titers, {\color{blue}$\square$} IgM titers and \color{black}$\star$} virus titers. The set of parameters used in each simulation is given in Tables \ref{tab:SF_value} and \ref{tab:Estimation}, and the patient identification number on the top right of each plot. }
\label{fig:evo_revVA_est}
\end{figure*}

\begin{figure*}
\centering
$f_{ij}$\hspace{3cm}$\tilde{\nu}_{ij}$\hspace{3cm}$\delta$\hspace{3cm}$\gamma_Y$\hspace{3cm}$t_V$\\
    \includegraphics[scale=0.46,trim=0cm 0cm 0cm 0.7cm,clip ]{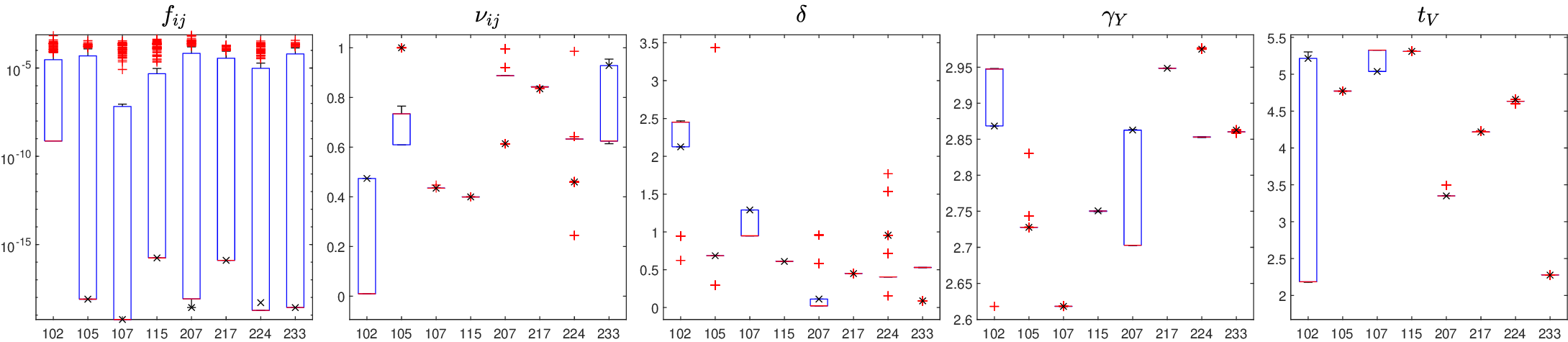}
    \caption{Boxplots of estimated model parameters $f_{ij}$, $\tilde{\nu}_{ij}$, $\delta$, $\gamma_Y$ and virus offset $t_V$ for each patient. The results comprise a set of 200 parameters obtained after running the Genetic Algorithm (GA) for 50 generations. The set of parameters corresponding to the best solution of the GA is represented by the symbol {\color{black}$\times$}, and the corresponding values are given in Table \ref{tab:SF_value}.}\label{fig:boxplot}
\end{figure*}

\subsection{Fitting dengue data}\label{FittData}

Here we compare the numerical simulations with data from measurements of virus and antibody titers collected sequentially during infection of dengue patients, Figure \ref{fig:evo_revVA_est}. These data were previously published in \cite{clapham_modelling_2016}. We selected part of the data corresponding to secondary infection by DENV-2 and, among these, we considered only eight individuals for whom the size of the temporal series was long enough to permit estimating the parameters related to the force of infection. To compare the data with the numerical simulations, the viral load data are offset by an elapsed time $t_V$ which corresponds to the incubation period of the dengue virus of 2-6 days, \cite{clapham_modelling_2016, Guzman_2016}. IgG and IgM antibody levels are multiplied by scaling factors, $SF_{IgG}$ and $SF_{IgM}$, because the ELISA assay cannot measure them directly, \cite{clapham_modelling_2016}. In addition, the IgG/IgM measurements are shifted by $t_{IgG}$ and $t_{IgM}$ to take into account the time delay between exposure to the antigen and the development of an effective immune response leading to clearance of the virus, \cite{ghosh_within_2021}. As this is a secondary, the IgG titers correspond to the antibodies generated by the first infection and the IgM titers to those generated by the second infection. The scaling factors $SF_{IgG}$, $SF_{IgM}$, and the antibody offsets $t_{IgG}$ and $t_{IgM}$ are defined by hand. Table \ref{tab:Parameters_value} shows the values of the parameters used in the simulations. Four parameters of the model $f_{ij}$, $\tilde{\nu}_{ij}$, $\delta$, $\gamma_Y$ plus the virus offset $t_V$ are estimated by a Genetic Algorithm (GA). For each individual patient, the best solution corresponds to the minimum of the function $\sum_{i=1}^{n_m}|\left(\log_{10}(V(t_i)+10)-\log_{10}(V_i+10)\right)|/\log_{10}(\sigma^2_V+10)$, where $\sigma^2_V$ is the variance of the viral load measurements, $V_i$ the i$^{th}$ measurement, $n_m$ the number of measurements for the patient, $V(t_i)$ the virus load given by the simulation at time $t_i$ The results are plotted in Figure \ref{fig:evo_revVA_est}. The parameter ranges used in the GA are as follows $f_{ij}\in [0,1]$, $\tilde{\nu}_{ij}\in[0,1]$, $\delta\in[0,5]$, $\gamma_Y\in[2.5,3]$ and $t_V\in[0,10]$. In addition, the other parameters used to run the GA are $200,\,0.1$, respectively, the size of the GA population and the crossover rate. The mutation function randomly generates directions that are adaptive with respect to the last successful or unsuccessful generation. The code was run in Matlab. Tables \ref{tab:SF_value}, \ref{tab:Estimation} and Figure \ref{fig:boxplot} summarize the scaling factors, antibody shifts, and estimated model parameters as well as virus offsets.  The main objective is either to validate the model and to discuss the heterogeneity of patients's viral kinetics. We note that our curves have the same shape as the data, Figure \ref{fig:evo_revVA_est}. In addition, parameter estimation provides very low values of $f_{ij}$, Figure \ref{fig:boxplot} and Tables \ref{tab:SF_value}, \ref{tab:Estimation}, showing that the two antibodies are in very strong competition (as the fraction of simultaneous binding is very low), which contributes to high virus levels. Whereas for patient 102, the values of $\tilde{\nu}_{ij}$ range from 0 to 0.5, for patient 233 they range from 0.6 to 1. In both cases, the relative viral activity due to synergistic binding is low, but more so for the first patient. In particular, the set of parameters corresponding to the best solution of the GA is represented by the symbol {\color{black}$\times$} in Figure \ref{fig:boxplot}. The model assumes a single immunity variable that starts immediately after virus exposure. But, natural infection by dengue virus triggers a fast antiviral innate immune response and a later adaptive immune response that resolves infection and leads to long-term immunological memory.

\section{Discussion}

We developed a formula that gives the virus activity when bound to two antibodies, using the principle of chemical reactions. It generalizes the work of \cite{einav_when_2020} by assuming that the antibodies can bind to multiple sites on a receptor. We have thus been able to show that this antigen-antibody complex can neutralize or enhance virus activity, given the interaction - competition or cooperation - between antibodies. The viral activity was then integrated into a model that emulates the interaction between healthy target cells, infected cells, viruses, and antibodies. This allows us to study the dynamic process of immune response when challenged by sequential homologous or heterologous viruses. We highlighted two scenarios: a mixture of two neutralizing antibodies (Figure \ref{fig:evo_rev1}) and a mixture of enhancing and neutralizing antibodies (Figure \ref{fig:evo_rev2}). Other situations not considered in this study, such as the mixture of two types of enhancing antibodies, could be envisaged. This can happen, for example, when using therapeutic antibodies. In this case, $p_i$, $i=1,2$, can also be a function of time. We could also generalize the model to the case where there are more than two antibody types, as in polyclonal antibody therapy, but the number of parameters would increase rapidly, making parametrization, analysis, and interpretation of the model a challenge. By keeping the model as simple as possible, we were able to explore it analytically and numerically and assert that virus neutralization is associated with purely independent binding (Figures \ref{fig:evo_rev1}(d) and \ref{fig:evo_rev2}(d)), while enhanced virus activity is expected when purely competitive binding occurs (Figure \ref{fig:evo_rev2}(b)). This opens up promising modeling prospects for the study of second viral infections, vaccine efficacy, or antibody treatment. Other potential applications, such as transcriptional gene regulation, \cite{cambon_thermodynamic_2022}, or any field where there is an interaction between activation and inhibition processes, could be favored by this type of new modeling. If the condition \eqref{Hyp:Reduction} and/or the assumption in Proposition \ref{P2} are relaxed, then we could have several endemic equilibrium points even in the parameter space where $\mathcal{R}_0<1$. It would be interesting to study their biological significance and asymptotic behavior, given that the latter depends on initial conditions such as the size of the virus inoculum and the amount of each antibody present in the mixture. Another challenge for our model is to confront it with biological data for dengue, influenza, Covid-19, or HIV. Preliminary tests carried out here on data from dengue patients are promising. The data concern secondary heterologous infection, and cover 5 or 6 days after the onset of symptoms; they therefore start around the time when virus titers reach their peak. For each patient, the model captures the range of virus titers and their rapid decrease after the peak. Rapid increases in IgG and IgM titers are also observed (Figure \ref{fig:evo_revVA_est}). Furthermore, for the set of data explored here, the antigen-antibody complex formed includes antibodies with a very low synergistic neutralizing effect (Figure \ref{fig:boxplot}). This corresponds to high viral activity. Finally, the estimated parameters show that the mortality of infected cells $\gamma_Y$ is almost twice as high as that of susceptible cells and that the heterogeneity of patients' viral kinetics is associated with variability of antibody responses between individuals, measured by the parameters $f_{ij}$, $\tilde{\nu}_{ij}$ and $\delta$ (Table \ref{tab:Estimation}). This is in agreement with previous works showing that antibody kinetics shape virus dynamics acting either on free virus or infected cells \cite{clapham_modelling_2016}. Other mechanisms not considered here, such as cell-mediated immune responses, may also be associated with the lysis of virus-infected cells. In addition, although there are only four antigenically distinct serotypes of the dengue virus, each serotype is divided into several genotypes. This factor, combined with genetics and the immune status of the host, may explain the kinetics of viral load and antibodies, \cite{Guzman_2016}.
\end{multicols}
\begin{small}
\begin{table}[ht] 
    \centering{
      \caption{Parameter of the mathematical model with their descriptions, values and units. Some parameters are varied among simulations and are distinguished: Figure \ref{fig:evo_rev1} $\square$, Figure \ref{fig:evo_rev2} $\boxtimes$, and Figure \ref{fig:evo_revVA_est} $\boxdot$. The baseline parameters are taken from the literature (see \cite{adimy_maternal_2020, camargo_modeling_2021}, for further details).}
        \label{tab:Parameters_value}  

\begin{tabular}{p{0.1\textwidth}p{0.5\textwidth}p{0.3\textwidth}}
     Parameter& Description& Value and Units \\\hline\hline
     $n_1,n_2$& Number of epitopes that can be recognized by $A_1$, $A_2$&3\\
     $\Omega$& Production rate of susceptible target cells&  $6.7 \times 10^3$ cells ml$^{- 1}$ days$^{-1}$\\ 
     $\Phi$& Production rate of virus by infected cells& $1 \times 10^4$ RNA copies cells$^{-1}$ days$^{-1}$\\
    $\beta$& Virus transmission rate &$8.5\times 10^{-10}$ ml RNA copies$^{-1}$ days$^{-1}$\\    
    $\gamma_X$& Natural mortality rate of target cells &$\log(2)/5$ days$^{-1}$\\
    $\gamma_Y$& Natural mortality rate of infected cells &$\log(2)$  days$^{-1}$ ; $\boxdot$\\
$\gamma_{V}$& Natural mortality rate of virus&$24\log(2)/2.7$ days$^{-1}$\\
$\gamma_{A_1}\,\gamma_{A_2}$& Natural mortality rate of antibodies &$\log(2)/40$ days$^{-1}$\\
$\alpha_1,\,\alpha_2$& Production rate of antibodies $A_1$, $A_2$ by infected cells&$0.4,\, 0.44$ mol cells$^{-1}$ days$^{-1}$\\ 
$\Lambda_1,\, \Lambda_2$& Natural production rate of antibodies $A_1$, $A_2$ &$7.5\times 10^3,\, 0$ mol ml$^{-1}$ days$^{-1}$\\
$(\xi_{11},\xi_{12},\xi_{13})$& Relative activity of the virus-antibody $A_1$ complex & $(1,0.95,0.8)$ $\square$;  $(3,3,3)$ $\boxtimes$\\
$(\xi_{21},\xi_{22},\xi_{23})$& Relative activity of the virus-antibody $A_2$ complex & $(1,0.95,0.8)$ $\square$; $(1/3,1/3,1/3)$ $\boxtimes$\\
$\tilde{\nu}_{ij}$& Ratio between modified and normal relative activities of the virus-antibody complex &0.65 $\square$; 0.4 $\boxtimes$ ; $\boxdot$ \\
 $K_{D_1},\, K_{D_2}$& Dissociation constants &0.8\\
 $\kappa_{ij}$& Ratio between modified and normal dissociation constants&1\\
 $\delta$& Rate of free virus loss by other means than natural mortality&0.8 days$^{-1}$ ; $\boxdot$\\
$f_{ij}$& Fraction of simultaneous binding of antibodies $A_1$, $A_2$ & $0,\,10^{-6},\,1$ $\square$ ; $0,\,10^{-3},\,1$  $\boxtimes$ ; $\boxdot$  \\ 
\hline
\end{tabular}}
\end{table}
\end{small}

\begin{table}[ht]
    \centering
         \caption{Fraction of simultaneous binding $f_{ij}$, ratio between modified and normal relative activities $\tilde{\nu}_{ij}$, rate of loss of free viruses by means other than natural mortality $\delta$,  natural mortality of infected cells $\gamma_Y$, time offset values $t_V$, $t_{IgG}$ and $t_{IgM}$, and scaling factors $SF_{IgG}$ and $F_{IgM}$, for each individual, used for Figure \ref{fig:evo_revVA_est}. The values of $f_{ij},\, \tilde{\nu}_{ij}$, $\delta$, $\gamma_Y$ and $t_V$ are the results of the Genetic Algorithm.}
        \label{tab:SF_value}    
\begin{tabular}{cccccccccc}
     Patient ID&$f_{ij}$&$\tilde{\nu}_{ij}$&$\delta$&$\gamma_Y$&$t_V$&$SF_{IgG}$&$t_{IgG}$&$SF_{IgM}$&$t_{IgM}$  \\\hline\hline
     102&0&0.47&2.13&2.87&5.2& $1.83\times 10^4 $&1&$3.88\times 10^4$&0.5\\
     105&$8.01\times10^{-19}$&1&0.69&2.73 &4.77&2.00$\times 10^4$&-0.5&$2.36\times 10^4$&-0.25\\
     107&$5.65\times10^{-20}$&0.43&1.29&2.62&5.04&$1.75\times 10^4$&0&$8.96\times 10^4$&0.5\\
     115&$1.74\times10^{-16}$&0.40&0.61&2.75&5.31&$1.92\times 10^4$&-2&$  1.62\times 10^4$&-1.5\\
     207&$2.68\times10^{-19}$&0.61&0.11&2.86&3.35&$3.07\times 10^4$&-0.5&$ 2.00\times 10^4$&1\\
     217&$1.28\times10^{-16}$&0.83&0.45&2.95&4.22&$2.50\times 10^4$&0&$2.64\times 10^4$&0\\
     224&$5.08\times10^{-19}$&0.46&0.96&2.98&4.66&$2.23\times 10^4$&-1&$2.46\times 10^4$&0\\
     233&$2.68\times10^{-19}$&0.93&0.09&2.86&2.28&$2.39\times 10^4$&-1&$3.46\times 10^4$&-1\\\hline 
\end{tabular} 
\end{table}
\begin{table}[ht]

    \centering
        \caption{Summary of the parameter values obtained by the Genetic Algorithm (GA) (see Table \ref{tab:SF_value}). In addition to parameter descriptions, we provide median, minimum and maximum values. }
        \label{tab:Estimation}  
        
\begin{tabular}{lll}
      Parameter& Description& Median [min, max]   \\\hline\hline
      $f_{ij}$& Fraction of simultaneous binding&$3.88\times10^{-19}$[$0,\, 1.74\times10^{-16}$] \\ 
       $\tilde{\nu}_{ij}$& Ratio between modified and normal relative activities&0.54 [0.39, 1]\\
     $\delta$& Rate of loss of free viruses by means other than natural mortality &0.65 [0.0882, 2.13] days$^{-1}$\\ 
    $\gamma_Y$& Natural mortality of infected cells &2.86 [2.61, 2.98] days$^{-1}$\\
  $t_V$&Time between onset of infection and first sampling &4.71 [2.27, 5.31] days
     \\\hline 
\end{tabular}
\end{table}

\newpage

\appendix

\section{Appendix}\label{Appendix1}

\subsection{Proofs of Theorems }\label{ProofsTh}

\begin{proof}[\proofname\ of Theorem \ref{T1}]
Let $J(P_0)$ be the Jacobian matrix associated with the system \eqref{Newmodel2}, evaluated at the equilibrium point $P_0$. Then, we have
\begin{equation*}
    J(P_{0}) = 
    \left(\begin{array}{ccccc}
        -\gamma_{A_1} & 0 & 0 & p_{1}'(0) & 0\\
        0 & -\gamma_{A_2} & 0 & p_{2}'(0) & 0 \\
        0 & 0 & -\gamma_{X} & 0 & -\beta G^{\star} X^{\star}\\
        0 & 0 & 0 & -\gamma_{Y} & \beta G^{\star} X^{\star} \\
        0 & 0 & 0 & \Phi & - \left(\gamma_V+\delta\right)
    \end{array}\right),
\end{equation*}
with $A_{1}^{\star} = \dfrac{p_{1}(0) }{\gamma_{A_1}}$, $A_{2}^{\star}=\dfrac{p_{2}(0) }{\gamma_{A_2}}$, $X^{\star} = \dfrac{\Omega}{\gamma_{X}}$ and $G^{\star}=G\left(A_{1}^{\star},A_{2}^{\star}\right)$. Then, the characteristic equation of $J(P_{0})$ is
\begin{equation*}
        \left|\begin{array}{ccccc}
        \lambda + \gamma_{A_1} & 0 & 0 & -p_{1}'(0) & 0\\
        0 & \lambda + \gamma_{A_2} & 0 & -p_{2}'(0) & 0 \\
        0 & 0 & \lambda +\gamma_{X} & 0 & \beta G^{\star} X^{\star}\\
        0 & 0 & 0 & \lambda + \gamma_{Y} & - \beta G^{\star} X^{\star} \\
        0 & 0 & 0 & -\Phi & \lambda +\left(\gamma_V+\delta\right)
    \end{array}\right| = 0.
\end{equation*}
The eigenvalues of $J(P_0)$ are $\lambda = -\gamma_{A_1}$, $\lambda = -\gamma_{A_2}$, $\lambda = -\gamma_X$  and the roots of the polynomial
\begin{equation*}
    (\lambda +\gamma_{Y})(\lambda  + \gamma_V+\delta) - \Phi \beta G^{\star} X^{\star}= 0.
\end{equation*}
We obtain the following polynomial
\begin{equation}
    \lambda^{2} + \left( \gamma_{V} +\delta+ \gamma_{Y} \right)\lambda + \gamma_{Y}\left(\gamma_V+\delta\right)- \Phi \beta G^{\star} X^{\star}= 0.
    \label{poly}
\end{equation}
According to the Routh-Hurwitz criterion, a polynomial of degree two has roots with a negative real part if and only if the coefficients $a_1$ and $a_2$ of the polynomial $\lambda^{2} + a_{1}\lambda + a_{2}$ are positive. For the polynomial \eqref{poly}, we have
\begin{eqnarray*}
    a_{1} &=&  \gamma_{V}+\delta+ \gamma_{Y} > 0,\\
    a_{2} &=& \gamma_{Y} \left(\gamma_V+\delta\right)-  \Phi \beta G^{\star} X^{\star},\\ &=& \gamma_{Y}\left(\gamma_V+\delta\right)( 1 - \mathcal{R}_{0}).
\end{eqnarray*}
Then, $a_{2} > 0$ if and only if $\mathcal{R}_{0} < 1.$
Therefore, the roots of the polynomial \eqref{poly} have negative real parts if and only if $\mathcal{R}_{0} < 1$. We conclude that the disease-free equilibrium $P_{0}$ is locally asymptotically stable if $\mathcal{R}_{0} < 1$, and unstable if $\mathcal{R}_{0} > 1$.
\end{proof}

\begin{proof}[\proofname\ of Theorem \ref{T2}]
Again, we take System \eqref{Newmodel2GlobStab}. Let first consider the case $\mu=-1$. Then, we have $\dfrac{\Omega}{\gamma} +\mu Z=X+Y\leq \dfrac{\Omega}{\gamma}$. Therefore, System \eqref{Newmodel2GlobStab} satisfies the following inequalities on $\mathbb{D} \times \mathbb{R}^3_+$
\begin{equation*}
	\left\{\begin{aligned}
		\dfrac{dB_{1}}{d t} & \leq p'_{1}(0)Y - \gamma_{A_1}B_{1}, \\
		\dfrac{dB_{2}}{d t} & \leq p'_{2}(0)Y - \gamma_{A_2}B_{2}, \\
		\dfrac{d Z}{d t} & =  - \gamma Z , \\
		\frac{d Y}{d t} &\leq\beta G\left(\dfrac{p_1(0)}{\gamma_{A_1}},\dfrac{p_2(0)}{\gamma_{A_2}}\right)\dfrac{\Omega}{\gamma}V - \gamma Y, \\
		\dfrac{d V}{d t} & =\Phi Y - \left(\gamma_V+\delta\right)V.
	\end{aligned}\right.
	\label{Newmodel2GlobStab:Compar}
\end{equation*}
Then, System \eqref{Newmodel2GlobStab} can be compared to the following linear system 
\begin{equation}\label{Newmodel2GlobStab:Compar Lin}
\left\{\begin{aligned}
		\dfrac{dB_{1}}{d t} & = p'_{1}(0)Y - \gamma_{A_1}B_{1}, \\
		\dfrac{dB_{2}}{d t} & = p'_{2}(0)Y - \gamma_{A_2}B_{2}, \\
		\dfrac{d Z}{d t} & =  - \gamma Z , \\
		\frac{d Y}{d t} &=\beta G\left(\dfrac{p_1(0)}{\gamma_{A_1}},\dfrac{p_2(0)}{\gamma_{A_2}}\right)\dfrac{\Omega}{\gamma}V - \gamma Y, \\
		\dfrac{d V}{d t} & =\Phi Y - \left(\gamma_V+\delta\right)V.
	\end{aligned}\right.
\end{equation}
This last system has the following characteristic equation
\begin{equation*}         
        \left|\begin{array}{ccccc}
        \lambda + \gamma_{A_1} & 0 & 0 & -p_{1}'(0)  & 0\\
        0 & \lambda + \gamma_{A_2} & 0 & -p_{2}'(0) & 0 \\
        0 & 0 & \lambda +\gamma & 0 &0\\
        0 & 0 & 0 & \lambda + \gamma & S\\
        0 & 0 & 0 & -\Phi & \lambda +\gamma_V+\delta
    \end{array}\right| = 0.
\end{equation*}
with $S:=-\beta G\left(p_{1}(0)/\gamma_{A_1},p_{2}(0)/\gamma_{A_2}\right)\Omega/\gamma$.
It is equivalent to 
\begin{equation}
(\lambda+\gamma_{A_1})(\lambda+\gamma_{A_2})(\lambda+\gamma)\times\left[\lambda^2+\lambda(\gamma+\gamma_V+\delta)+\gamma\left(\gamma_V+\delta\right)(1-\mathcal{R}_0)\right]=0.
\end{equation}
The eigenvalues are $\lambda=-\gamma_{A_1}<0$, $\lambda=-\gamma_{A_2}<0$, $\lambda=-\gamma<0$ and the roots of the polynomial
\[\lambda^2+(\gamma+\gamma_V+\delta)\lambda+\gamma\left(\gamma_V+\delta\right)(1-\mathcal{R}_0).\]
Using the Routh-Hurwitz criteria, we find that the roots of this last polynomial have negative real parts if and only if $\mathcal{R}_0<1$. Thus, we obtain the global asymptotic stability of the trivial equilibrium of the linear system \eqref{Newmodel2GlobStab:Compar Lin}. We now use the following comparison result to conclude (Lemma 1 of \cite{camargo_modeling_2021} and \cite{Kirkilionis_comparison_2004}). 
\begin{lemma}\label{L2}
Consider two differential systems $x' = h(x)$ and $y' = g(y)$ given on an invariant subset $U$ of $\mathbb{R}^k$, $h$, $g : U \to \mathbb{R}^k$ are locally Lipschitz functions. Then, the following two conditions are equivalent: 
\begin{enumerate}
    \item  For each $x_0, y_0 \in U$ the inequality $x_0 \leq y_0$ implies $x(t) \leq y(t)$ for all $t \geq 0$, where $x'(t) = h (x(t))$ and $y'(t) = g(y(t))$, $t \geq 0$, with $x(0) = x_0$, $y(0) = y_0$.
    \item For all $ i=1, \dots , k$, the inequality 
    \begin{equation*}
       h_i(x_1, \dots,x_{i-1},x_i,x_{i+1},\dots, x_n)\quad\leq \quad g_i(\Bar{x}_1,\dots,\Bar{x}_{i-1},\Bar{x}_i,\Bar{x}_{i+1},\dots,\Bar{x}_n)
    \end{equation*}
    holds whenever $x_j\leq \Bar{x}_j$, 
    for all $j \neq i$ and $x_i=\bar{x}_i$.
\end{enumerate}\label{lem:comparison}
\end{lemma}
As our corresponding functions $f$ and $g$ are of class $\mathcal{C}^1$, they are locally Lipschitz and we can apply Lemma \ref{L2}. Therefore, if we chose the same initial condition on $\mathbb{D} \times \mathbb{R}^3_+$, we obtain a solution of System \eqref{Newmodel2GlobStab} between $0_{\mathbb{R}^5}$ and the solution of System \eqref{Newmodel2GlobStab:Compar Lin}. Hence, the disease-free steady-state $P_0$ is globally asymptotically stable on the subset $\mathbb{D} \times \mathbb{R}^3_+$.

We continue the proof of Theorem \ref{T2} by now considering the case $\mu=1$. Using the equation of $Z$, we can show for all $\epsilon>0$ and all $Z_0>0$, the existence of $T_\epsilon \geq 0$, we can take $T_\epsilon=\text{max}\left\{0, \dfrac{1}{\gamma}\text{ln}\left(\dfrac{Z_0}{\epsilon}\right)\right\}$, such that $0<Z(t)<\epsilon$, for all $t>T_\epsilon$. We then compare System \eqref{Newmodel2GlobStab}, for $t>T_\epsilon$, to the following linear system 
\begin{equation}
	\left\{\begin{aligned}
		\dfrac{dB_{1}}{d t} & =p'_{1}(0)Y- \gamma_{A_1}B_{1}, \\
		\dfrac{dB_{2}}{d t} & = p'_{2}(0)Y - \gamma_{A_1}B_{2}, \\
		\dfrac{d Z}{d t} & =  - \gamma Z , \\
		\frac{d Y}{d t} &=\beta G\left(\dfrac{p_{1}(0)}{\gamma_{A_1}},\dfrac{p_{2}(0)}{\gamma_{A_2}}\right)\left(\dfrac{\Omega}{\gamma}+\epsilon\right)V - \gamma Y, \\
		\dfrac{d V}{d t} & = \Phi Y - \left(\gamma_V+\delta\right)V.
	\end{aligned}\right.
	\label{Newmodel2GlobStab:Compar Lin2}
\end{equation}
Its characteristic equation is given by
\begin{equation}
0=(\lambda+\gamma_{A_1})(\lambda+\gamma_{A_2})(\lambda+\gamma)\times \left[(\lambda^2+\lambda(\gamma+\gamma_V)+\gamma\gamma_V(1-\mathcal{R}_\epsilon)\right],
\end{equation}
where 
\[\mathcal{R}_\epsilon:=\dfrac{\Phi \beta}{\gamma\left(\gamma_V+\delta\right)}\left(\dfrac{\Omega}{\gamma}+\epsilon\right)G\left(\dfrac{p_{1}(0)}{\gamma_{A_1}},\dfrac{p_{2}(0)}{\gamma_{A_2}}\right).\]
Under the condition $\mathcal{R}_0<1$ and by choosing $\epsilon>0$ small enough, we have $\mathcal{R}_\epsilon<1$. Then, we obtain the global asymptotic stability of the trivial solution of the linear system \eqref{Newmodel2GlobStab:Compar Lin2}. So, using again Lemma \ref{lem:comparison}, we conclude the global asymptotic stability of the disease-free steady-state $P_0$ on the set $\mathbb{D} \times \mathbb{R}^3_+$. So far, we have considered the solutions of System \eqref{Newmodel2} that start in the set $\mathbb{D} \times \mathbb{R}^3_+$. Now consider a nonnegative solution such that $A_1(0)<\frac{p_{1}(0)}{\gamma_{A_1}}$ and/or $A_2(0)<\frac{p_{2}(0)}{\gamma_{A_2}}$, which means that it does not start in the set $\mathbb{D} \times \mathbb{R}^3_+$. We have two different cases. \\
Firstly, suppose that we have the existence of $\Bar{t}_i>0$, $i=1,2$, such that $A_i(\Bar{t}_i)=\frac{p_{i}(0)}{\gamma_{A_i}}$ and $A_i(t)<\frac{p_{i}(0)}{\gamma_{A_i}}$ for all $t<\Bar{t}_i$. From the equation of $A_i$ in System \eqref{Newmodel2} and under \eqref{Hyp:Condition_f}, we can see that $A'_i(t)\geq p_{i}(Y(t))-p_{i}(0)\geq 0$, for $t\leq \Bar{t}_i$. This means that $A_i$ is an increasing function on $[0,\Bar{t}_i]$. Hence, the solution enters the subset $\mathbb{D} \times \mathbb{R}^3_+$ at time $\Bar{t}=\underset{i=1,2}{max}(\Bar{t}_i)$. As $\mathbb{D} \times \mathbb{R}^3_+$ is invariant for the system \eqref{Newmodel2}, the solution tends to $P_0$, because $P_0$ is globally asymptotically stable on the set $\mathbb{D} \times \mathbb{R}^3_+$.\\
Secondly, we suppose that there is no $\Bar{t}_i$ as described above (for $i=1$ or/and $i=2$). We consider that for all $t\geq 0$, $A_i(t)<\frac{p_{i}(0)}{\gamma_{A_i}}$. Then, $A_i'(t)>0$. This means that $A_i$ is increasing on $[0,+\infty)$ and admits a limit $l_i:=\underset{t\to+\infty}{\text{lim}}A_i(t)\leq\frac{p_{i}(0)}{\gamma_{A_i}}$. Furthermore, $\underset{t \to +\infty}{\text{lim}}A'_i(t)=0$. Then, we can see that 
\[\underset{t\to+\infty}{\text{lim}}p_i(Y(t))=\gamma_{A_i} l_i \leq p_i(0).\]
But, according \eqref{Hyp:Condition_f'(0)}, $p_i$ ($i=1$ or $i=2$) is an increasing function in a neighbourhood of $0$. So, we conclude that 
\[\underset{t\to+\infty}{\text{lim}}Y(t)=0 \quad \text{and} \quad \underset{t\to+\infty}{\text{lim}}A_i(t)=l_i=\dfrac{p_i(0)}{\gamma_{A_i}}.\]
We have shown that $\mathbb{D} \times \mathbb{R}^3_+$ is globally attractive for System \eqref{Newmodel2}. We can now conclude that the disease-free steady-state $P_0$ is globally asymptotically stable.
\end{proof}

\begin{proof}[\proofname\ of Theorem \ref{Thm:exist_P*}]
	An endemic equilibrium $P^{\star}=(A_{1}^{\star},A_{2}^{\star},X^{\star},Y^{\star},V^{\star})$ satisfies
\begin{equation}
		\left\{\begin{array}{ll}
			\gamma_{A_1} A_{1}^{\star}
			& = p_{1} (Y^{\star} ),\\
			\gamma_{A_2} A_{2}^{\star} 
			& =p_{2} (Y^{\star} ),\\
			\gamma_{X} X^{\star}+  \beta {G}\left(A_{1}^{\star}, A_{2}^{\star}\right)V^{\star}X^{\star} & = \Omega, \\
			\beta  {G}\left(A_{1}^{\star}, A_{2}^{\star}\right)V^{\star}X^{\star} & =\gamma_{Y} Y^{\star}, \\
			 \left(\gamma_V+\delta\right)V^{\star} & =\Phi Y^{\star}.
		\end{array}\right.
		\label{endem}
	\end{equation}
Subtracting the third from the fourth equation of the \eqref{endem} system, we get,
\begin{eqnarray*}
    \Omega - \gamma_{X}X^{\star} - \gamma_{Y}Y^{\star} = 0,
\end{eqnarray*} which results in
\begin{eqnarray}
    X^{\star} = \dfrac{\Omega - \gamma_{Y}Y^{\star}}{\gamma_{X}}, \quad \text{with} \quad 0 < Y^{\star} \leq\dfrac{\Omega}{\gamma_{Y}}.
\end{eqnarray}
From the first equation of the system, we obtain,
\begin{equation*}
    A_{1}^{\star} = \dfrac{1}{\gamma_{A_1}}p_{1}(Y^{\star}).
\end{equation*}
From the second equation of the system \eqref{endem}, we get
\begin{equation*}
    A_{2}^{\star} = \dfrac{1}{\gamma_{A_2}}p_{2}(Y^{\star}).
\end{equation*}
We can see that $\left(A_{1}^{\star},A_{2}^{\star}\right) \in \mathbb{D}$, then $ P^{\star} \in \mathbb{D}\times \mathbb{R}^3_+$. From the last equation of System \eqref{endem}, we get
\begin{equation}
    V^{\star} = \dfrac{\Phi Y^{\star}}{\gamma_V+\delta}.
\end{equation}
Substituting the equations of $X^{\star}$ and $V^{\star}$ into the fourth equation of the system \eqref{endem} we get
\begin{eqnarray*}
    	\beta G(A_{1}^{\star},A_{2}^{\star})\left(\dfrac{\Phi Y^{\star}}{ \gamma_V+\delta}\right)\dfrac{\Omega - \gamma_{Y}Y^{\star}}{\gamma_{X}} =  \gamma_{Y}Y^{\star}.
\end{eqnarray*}
As we are looking for $Y^{\star}>0$, the previous equation is equivalent to
\begin{eqnarray*}
    	\dfrac{\beta \Phi}{\gamma_{Y}\gamma_X}\left(\Omega - \gamma_{Y}Y^{\star}\right)G(A_{1}^{\star},A_{2}^{\star})\left(\dfrac{1}{ \gamma_V+\delta}\right)  =1  
\end{eqnarray*}
We consider the function $I : \left[0,\dfrac{\Omega}{\gamma_{Y}} \right] \rightarrow \left[ 0, +\infty\right) $ defined by
\begin{eqnarray*}
    I(Y^{\star}) = \dfrac{1}{\gamma_{Y}}\times\dfrac{\Omega - \gamma_{Y}Y^{\star}}{\gamma_{X}}\times\dfrac{\Phi}{ \gamma_V+\delta}\times\beta G\left(\dfrac{p_{1}(Y^{\star})}{\gamma_{A_1}},\dfrac{p_{2}(Y^{\star})}{\gamma_{A_2}}\right).
\end{eqnarray*} 
The problem to solve is the following 
\begin{equation*}
I(Y^{\star}) = 1, \quad \text{with} \;\; 0 \leq Y^{\star} \leq \dfrac{\Omega}{\gamma_{Y}}.
\end{equation*}
We have
\begin{equation*}
    I(0) = \mathcal{R}_0 \quad \text{and} \quad  I\left(\dfrac{\Omega}{\gamma_{Y}}\right) = 0.
\end{equation*} 
Therefore, if $\mathcal{R}_0>1$, there is at least one positive solution $Y^{\star}$.
\end{proof}

\begin{proof}[\proofname\ of Theorem \ref{Thm:local_stab_P*}]
Let $A_1(t)=A_1^{\star}+a_1(t)$, $A_2(t)=A_2^{\star}+a_2(t)$, $X(t)=X^{\star}+x(t)$, $Y(t)=Y^{\star}+y(t)$, $V(t)=V^{\star}+v(t)$. The linearization of System \eqref{Newmodel2} around $P^{\star}$ is given by
     \begin{equation}
	\left\{\begin{aligned}
		\dfrac{da_{1}}{d t}  =&  p'_1(Y^{\star})y - \gamma_{A}a_{1}, \\
		\dfrac{da_{2}}{d t}  =& p'_2(Y^{\star})y - \gamma_{A}a_{2}, \\
		\dfrac{d x}{d t}  = &- \gamma x - \beta G(A_1^{\star},A_2^{\star})(V^{\star}x+X^{\star}v)\\&-\beta G_{A_1}(A_1^{\star},A_2^{\star})V^{\star}X^{\star}a_1-\beta G_{A_2}(A_1^{\star},A_2^{\star})V^{\star}X^{\star}a_2, \\
		\frac{d y}{d t}  =& -\gamma y+  \beta G(A_1^{\star},A_2^{\star})(V^{\star}x+X^{\star}v)\\&+\beta G_{A_1}(A_1^{\star},A_2^{\star})V^{\star}X^{\star}a_1+\beta G_{A_2}(A_1^{\star},A_2^{\star})V^{\star}X^{\star}a_2, \\
		\dfrac{d v}{d t}  =& \Phi y - \left( \gamma_V+\delta\right) v.
	\end{aligned}\right.
	\label{LinearNewmodel2}
\end{equation}
The characteristic equation associated with the previous system is 
\[\begin{vmatrix}
\lambda+\gamma_{A}&0&0&-\alpha_1&0\\
0&\lambda+\gamma_{A}&0&-\alpha_2&0\\
\beta_1V^{\star}X^{\star}&\beta_2V^{\star}X^{\star}&\lambda+\gamma +\beta_0V^{\star}&0&\beta_0X^{\star}&\\
-\beta_1V^{\star}X^{\star}&-\beta_2V^{\star}X^{\star}&-\beta_0V^{\star}&\lambda+\gamma &-\beta_0X^{\star}\\
0&0&0&-\Phi&\lambda+ \gamma_V+\delta
\end{vmatrix}=0,\]
 with
 \[\alpha_i=p'_i(Y^{\star})\geq 0, \quad \beta_i=\beta G_{A_i}(A_1^{\star},A_2^{\star}), \quad i=1,2 \quad \text{and} \quad \beta_0=\beta G(A_1^{\star},A_2^{\star})>0.\]
We add the third and the fourth lines that we substitute to line 4 and  the characteristic equation becomes 
\[\begin{vmatrix}
\lambda+\gamma_{A}&0&0&-\alpha_1&0\\
0&\lambda+\gamma_{A}&0&-\alpha_2&0\\
\beta_1V^{\star}X^{\star}&\beta_2V^{\star}X^{\star}&\lambda+\gamma +\beta_0V^{\star}&0&\beta_0X^{\star}&\\
0&0&\lambda+\gamma &\lambda+\gamma &0\\
0&0&0&-\Phi&\lambda+ \gamma_V+\delta
\end{vmatrix}=0.\] 
The subtraction of columns 3 and 4 that we substitute to column 4 gives,
\[\begin{vmatrix}
\lambda+\gamma_{A}&0&0&-\alpha_1&0\\
0&\lambda+\gamma_{A}&0&-\alpha_2&0\\
\beta_1V^{\star}X^{\star}&\beta_2V^{\star}X^{\star}&\lambda+\gamma+\beta_0V^{\star}&-\lambda-\gamma-\beta_0V^{\star}&\beta_0X^{\star}&\\
0&0&\lambda+\gamma&0&0\\
0&0&0&-\Phi&\lambda+\gamma_V+\delta 
\end{vmatrix}=0.\] 
Developing firstly from the fourth line and secondly from the first column we obtain,
\[-(\lambda+\gamma)\left( (\lambda+\gamma_{A})D_1+\beta_1V^{\star}X^{\star}D_2\right)=0,\]
with 
\[D_1=\begin{vmatrix}
\lambda+\gamma_{A}&-\alpha_2&0\\
\beta_2V^{\star}X^{\star}&-\lambda-\gamma-\beta_0V^{\star}&\beta_0X^{\star}&\\
0&-\Phi&\lambda+ \gamma_V+\delta
\end{vmatrix},\]
and
\[D_2=\begin{vmatrix}
0&-\alpha_1&0\\
\lambda+\gamma_{A}&-\alpha_2&0\\
0&-\Phi&\lambda+ \gamma_V+\delta
\end{vmatrix}.\]
Hence we get,
\begin{eqnarray*}
    D_1=&(\lambda+\gamma_{A})\left[(-\lambda-\gamma-\beta_0V^{\star})(\lambda+ \gamma_V+\delta)+\Phi \beta_0X^{\star}\right]+\alpha_2\beta_2V^{\star}X^{\star}(\lambda+ \gamma_V+\delta),
\end{eqnarray*}
and 
\[D_2=\alpha_1(\lambda+\gamma_{A})\left(\lambda+ \gamma_V+\delta\right).\]
Then, we get the following characteristic equation
\[-(\lambda+\gamma) (\lambda+\gamma_{A})\left(D_1+\alpha_1\beta_1V^{\star}X^{\star}\left(\lambda+ \gamma_V+\delta\right)\right)=0,\]
 $\lambda=-\gamma<0$ and $\lambda=-\gamma_A<0$ are two eigenvalues. The others eigenvalues are solution of 
 \[D_1+\alpha_1\beta_1V^{\star}X^{\star}\left(\lambda+ \gamma_V+\delta\right)=0.\]
This corresponds to 
 \begin{eqnarray*}
    (\lambda+\gamma_{A})\left[(\lambda+\gamma+\beta_0V^{\star})(\lambda+ \gamma_V+\delta)-\Phi \beta_0X^{\star}\right]-\alpha_1\beta_1V^{\star}X^{\star}\left(\lambda+ \gamma_V+\delta\right)-\alpha_2\beta_2V^{\star}X^{\star}(\lambda+ \gamma_V+\delta)=0.
 \end{eqnarray*}
 Developing this last equation, we obtain
 \begin{eqnarray*}
 &&\lambda^3+\lambda^2\left[\gamma_{A}+\gamma+\gamma_{V}+\delta+\beta_0V^{\star}\right]\\
 &+&\lambda \left[\gamma_{A}\left(\gamma+\gamma_{V}+\delta+\beta_0V^{\star}\right)+(\gamma_{V}+\delta)(\gamma+\beta_0V^{\star})-\Phi \beta_0X^{\star}-\alpha_1 \beta_1V^{\star}X^{\star}-\alpha_2\beta_2V^{\star}X^{\star}\right]\\
 &+&\gamma_{A}(\gamma_{V}+\delta)(\gamma+\beta_0V^{\star})-\gamma_A\Phi \beta_0X^{\star}-\alpha_1 \beta_1V^{\star}X^{\star}(\gamma_{V}+\delta)-\alpha_2 \beta_2V^{\star}X^{\star}(\gamma_{V}+\delta) \qquad\quad= 0.
  \end{eqnarray*}
 That we rewrite
\begin{equation}
    \lambda^3 +k_1\lambda^2+k_2 \lambda+k_3=0,
    \label{eq:characteristic}
\end{equation}
with
\begin{eqnarray*}
    k_1&=&\gamma_A + \gamma_V+\delta+\gamma+\beta_0V^{\star},\\
    k_2&=&\gamma_A\left( \gamma_V+\delta+\gamma+\beta_0V^{\star}\right)+( \gamma_V+\delta)(\gamma+\beta_0V^{\star})-\Phi \beta_0X^{\star}-(\alpha_1\beta_1+\alpha_2\beta_2)V^{\star}X^{\star},\\
    k_3&=&\gamma_A\left(( \gamma_V+\delta)(\gamma+\beta_0V^{\star})-\Phi  \beta_0X^{\star} \right) -(\alpha_1\beta_1+\alpha_2\beta_2)( \gamma_V+\delta)V^{\star}X^{\star}.
\end{eqnarray*}
According to the Routh-Hurwitz criteria, the endemic steady-state is locally asymptotically stable if and only if the following conditions on the coefficients of Equation \eqref{eq:characteristic} are satisfied
\[k_1>0,\quad k_3>0 \quad \text{and}\quad k_1k_2>k_3.\]
It is clear that we have always
\begin{eqnarray*}
     k_1&=&\gamma_A + \gamma_V+\delta+\gamma+\beta_0V^{\star}>0.
\end{eqnarray*}
We recall the equations that come from the steady state equation \eqref{endem}
\[\gamma_{A} A_{1}^{\star}= p_1(Y^{\star}), \quad \gamma_{A} A_{2}^{\star} = p_2(Y^{\star}), \quad \gamma X^{\star}+ \beta_0V^{\star}X^{\star}  = \Omega, \quad \beta_0V^{\star}X^{\star}  =\gamma Y^{\star} \quad \text{and} \quad \left(\gamma_V+\delta\right)V^{\star} =\Phi Y^{\star}. \]
Then, $k_3$ becomes
\begin{eqnarray*}
k_3&=&\gamma_A\left(( \gamma_V+\delta)(\gamma+\beta_0V^{\star})-\Phi  \beta_0X^{\star} \right) -(\alpha_1\beta_1+\alpha_2\beta_2)( \gamma_V+\delta)V^{\star}X^{\star},\\
&=& \gamma_A\Phi\dfrac{Y^{\star}}{V^{\star}}\left(\dfrac{\Omega}{X^{\star}}-\gamma\right)   -(\alpha_1\beta_1+\alpha_2\beta_1)\Phi Y^{\star}X^{\star}.
\end{eqnarray*}
$Y^{\star}>0$ and $V^{\star}>0$, and so $0<X^{\star}<\Omega/\gamma$. Therefore, under the condition \eqref{Condition Derivative G}, we have $\alpha_1 \beta_1+\alpha_2 \beta_2 \leq 0$. Then, we directly get $k_3>0$. Using the above equations, we can also rewrite $k_1$ and $k_2$ as follows 
\[k_1=\gamma_A+\Phi\dfrac{Y^{\star}}{V^{\star}}+\dfrac{\Omega}{X^{\star}},\]
\[k_2=\gamma_A \left(\Phi\dfrac{Y^{\star}}{V^{\star}}+\dfrac{\Omega}{X^{\star}}\right)+\Phi\dfrac{Y^{\star}}{V^{\star}}\left(\dfrac{\Omega}{X^{\star}}-\gamma\right)-(\alpha_1\beta_1+\alpha_2\beta_2)V^{\star}X^{\star}.\]
After some calculations, we obtain
\begin{eqnarray*}
    k_1k_2-k_3&=& \left(\gamma_A+\Phi\dfrac{Y^{\star}}{V^{\star}}\right) \left(\gamma_A\left(\Phi\dfrac{Y^{\star}}{V^{\star}}+\dfrac{\Omega}{X^{\star}} \right)+\Phi\dfrac{Y^{\star}}{V^{\star}}\left(\dfrac{\Omega}{X^{\star}}-\gamma\right)\right)\\
    &&+\dfrac{\Omega}{X^{\star}}\left(\gamma_A\dfrac{\Omega}{X^{\star}} +\Phi\dfrac{Y^{\star}}{V^{\star}}\left(\dfrac{\Omega}{X^{\star}}-\gamma\right)\right)+\gamma\gamma_A\Phi\dfrac{Y^{\star}}{V^{\star}}\\
    &&-(\alpha_1\beta_1+\alpha_2\beta_2)\left(\Omega V^{\star}+\gamma_AV^{\star}X^{\star}\right).
\end{eqnarray*}
Then, Condition \eqref{Condition Derivative G} implies that
\[k_2>0 \quad \text{and} \quad k_1k_2-k_3>0.\]
So, the endemic equilibrium $P^{\star}=(A_{1}^{\star},A_{2}^{\star},X^{\star},Y^{\star},V^{\star})$ is locally asymptotically stable.

\end{proof}

\subsection{Proofs of Propositions}\label{sec:ProofProp}
\begin{proof}[\proofname\ of Proposition \ref{P1}]
The regularity of the functions used in the right-hand side of System \eqref{Newmodel2} guarantees the existence and uniqueness of solutions on an interval $[0,T)$, with $T>0$. Let us see first that the solution is nonnegative on $[0,T)$. For all $t \in [0,T)$, whenever $(A_1,A_2,X,Y,V)\in \mathbb{R}^5_{+}$, the derivatives satisfy
\[\left.\dfrac{dA_1}{dt}\right|_{A_1=0} \geq p_1(0) \geq 0, \quad \left.\dfrac{dA_2}{dt}\right|_{A_2=0}\geq p_2(0) \geq 0,\quad\left.\dfrac{dX}{dt}\right|_{X=0} = \Omega>0, \quad\left.\dfrac{dY}{dt}\right|_{Y=0}\geq 0 \quad \text{and}\quad\ \left.\dfrac{dV}{dt}\right|_{V=0}\geq 0.\]

Hence, any solution of System \eqref{Newmodel2} that starts nonnegative remains nonnegative (see Theorem 3.4 in \cite{smith_introduction_2011} and Proposition
B.7 in \cite{smith_chermostat_1995} for more details).

Now, we prove that any solution of System \eqref{Newmodel2} is bounded on $[0,T)$. By adding the equations of $X$ and $Y$, we get 
\[\dfrac{d X}{d t}+\dfrac{d Y}{d t}\leq\Omega - \gamma_{m}(X+Y), \quad \text{with} \;\; \gamma_{m}=\text{min}\{\gamma_{X},\gamma_{Y}\}.\]
Then, 
\begin{equation*}
	\begin{array}{ll}
0\leq X(t)+Y(t)&\leq e^{-\gamma_{m}t}(X(0)+Y(0))+ \dfrac{\Omega}{\gamma_{m}}(1-e^{-\gamma_{m}t}),\\&\leq X(0)+Y(0)+ \dfrac{\Omega}{\gamma_{m}}.
		\end{array}
	\end{equation*}

So, 
\[\underset{t\to+\infty}{\text{lim sup}}(X(t)+Y(t))\leq X(0)+Y(0)+ \dfrac{\Omega}{\gamma_{m}}.\]
Hence, $X$ and $Y$ are bounded on the interval $[0,T)$. Similarly, we find
\begin{equation*}
	\begin{array}{ll}
&\underset{t\to+\infty}{\text{lim sup}}(A_1(t))\leq A_1(0)+\dfrac{p_1(\Bar{Y})}{\gamma_{A_1}}, \\& \underset{t\to+\infty}{\text{lim sup}}(A_2(t))\leq A_2(0)+\dfrac{p_2(\Bar{Y})}{\gamma_{A_2}},
		\end{array}
	\end{equation*}

and
\[\underset{t\to+\infty}{\text{lim sup}} (V(t))\leq V(0) +\dfrac{\Phi \Bar{Y}}{\gamma_V+\delta}, \quad \text{with} \;\; \Bar{Y}=\underset{0\leq s <T}{\text{sup}}(Y(s)).\]
Therefore the solution is bounded on $[0,T)$. This means that the solution is defined and bounded over the whole interval $[0,+\infty)$.
\end{proof}

\begin{proof}[\proofname\ of Proposition \ref{Prop:noEndemic}]
The condition \eqref{Hyp:Reduction} implies that 
\[I(Y^{\star}) \leq \dfrac{\Phi \beta \Omega}{\gamma_{Y}\gamma_{X}\left(\gamma_V+\delta\right)}G\left(\dfrac{p_{1}(0)}{\gamma_{A_{1}}},\dfrac{p_{2}(0)}{\gamma_{A_{2}}}\right)=\mathcal{R}_0<1,\]for all $Y^{\star}>0$.
Then, it is clear that there is no endemic steady-state.
\end{proof}

\begin{proof}[\proofname\ of Proposition \ref{P2}]
The hypothesis of Proposition \ref{P2} implies that the function $I : \left[0,\Omega/\gamma_{Y} \right] \to \left[ 0, +\infty\right)$ is decreasing with $I(0)=\mathcal{R}_0$ and $I\left(\Omega/\gamma_{Y}\right)=0$. Then, the equation $I(Y^{\star})=1$ has a positive solution if and only if $\mathcal{R}_0>1$.
\end{proof}

\subsection{Proof of Lemma}\label{sec:ProofLemma}
\begin{proof}[\proofname\ of Lemma \ref{Lem:D*Rinvariant}]
Under the condition \eqref{Hyp:gamma_X}, $\gamma:=\gamma_X=\gamma_Y$, we have
\[\dfrac{d}{dt}(X+Y)=\Omega-\gamma (X+Y)=\gamma\left(\dfrac{\Omega}{\gamma}-(X+Y)\right).\]
Let $B_1=A_1-\dfrac{p_{1}(0)}{\gamma_{A_1}}$, $B_2=A_2-\dfrac{p_{2}(0)}{\gamma_{A_2}}$ and
\begin{equation*}
    Z=\left\{\begin{aligned}
     &\dfrac{\Omega}{\gamma}-(X+Y), \quad\text{if}\quad X+Y\leq\dfrac{\Omega}{\gamma},\\
    &(X+Y) -\dfrac{\Omega}{\gamma},  \quad \text{if} \quad X+Y>\dfrac{\Omega}{\gamma}.
    \end{aligned}  \right.\\
\end{equation*}
It is clear that $X+Y$ cannot cross $\dfrac{\Omega}{\gamma}$. Then, with the new functions $B_1$, $B_2$ and $Z$, System \eqref{Newmodel2} becomes
\begin{equation}
	\left\{\begin{aligned}
		\dfrac{dB_{1}}{d t} & =p_{1}(Y)-p_{1}(0) - \gamma_{A_1}B_{1}, \\
		\dfrac{dB_{2}}{d t} & = p_{2}(Y)-p_{2}(0) - \gamma_{A_2}B_{2}, \\
		\dfrac{d Z}{d t} & =  - \gamma Z , \\
		\frac{d Y}{d t} & = \beta G\left(B_1+\dfrac{p_{1}(0)}{\gamma_{A_1}},B_2+\dfrac{p_{2}(0)}{\gamma_{A_2}}\right)\\
  &\times V\left(\dfrac{\Omega}{\gamma} +\mu Z - Y\right)- \gamma Y, \\
		\dfrac{d V}{d t} & = \Phi Y - \left(\gamma_V+\delta\right) V,
	\end{aligned}\right.
	\label{Newmodel2GlobStab}
\end{equation}
where 
\begin{equation*}
\left\{\begin{array}{lcl}
\mu=-1 & \text{if} & X+Y\leq\dfrac{\Omega}{\gamma}, \vspace{0.1cm}\\
\mu=1  & \text{if} & X+Y>\dfrac{\Omega}{\gamma}.
\end{array}\right.
\end{equation*}
If $\mu=1$, then $\dfrac{\Omega}{\gamma} +\mu Z=\dfrac{\Omega}{\gamma} +X+Y-\dfrac{\Omega}{\gamma} =X+Y\geq0$. \\
If $\mu=-1$, then $\dfrac{\Omega}{\gamma} +\mu Z=\dfrac{\Omega}{\gamma} -\left(\dfrac{\Omega}{\gamma} -(X+Y)\right)=X+Y\geq0$.\\ 
So, in all cases $\frac{\Omega}{\gamma} +\mu Z\geq 0$. Moreover, the equilibrium $P_0$ becomes $0_{\mathbb{R}^5}$ for System \eqref{Newmodel2GlobStab} and the set $\mathbb{D} \times \mathbb{R}^3_+$ corresponds now to the non-negative orthant $\mathbb{R}^5_+$. We have already proved in Proposition \ref{P1} that $Y$ and $V$ are nonnegative, and as the function $p_i$ is nondecreasing, we also have $B'_i(t)\geq -\gamma_{A_i}B_i(t)$, for all $t\geq 0$. Hence, any solution of System \eqref{Newmodel2GlobStab} that starts in $\mathbb{R}^5_+$ remains in $\mathbb{R}^5_+$. Then, the subset $\mathbb{D}\times \mathbb{R}^3_+$ is invariant under System \eqref{Newmodel2}.
\end{proof}

\enlargethispage{20pt}






\textbf{Acknowledgements.}CPF thanks grant \# 304007/2023-4, National Council for Scientific and Technological Development (CNPq). This work was supported by grants \# 2019/22157-5, S\~ao Paulo Research Foundation (FAPESP) and Capes 88881.878875/2023-01. MA and CDC thank the Inria international associate team MoCoVec and BIO-CIVIP, STIC AmSud 23-STIC-02.

\end{document}